\documentclass[11pt,journal, onecolumn]{./IEEEtran}
\renewcommand{\baselinestretch}{2.0} 
\usepackage[usenames]{color}
\usepackage{amsmath,amssymb,epsfig}
\usepackage{indentfirst}
\usepackage{cite}
\usepackage{subfigure}
\usepackage{multirow}

\newcommand{\blue}{\color{black}}
\newcommand{\black}{\color{black}}
\newcommand{\red}{\color{black}}

\begin{document}

\title{Tradeoff Analysis of Delay-Power-CSIT Quality of Dynamic BackPressure Algorithm for Energy Efficient OFDM Systems}

\author{\authorblockN{Vincent K. N. Lau and Chung Ha Koh}\\\authorblockA{Dept. of Electrical and Computer Engineering
\\Hong Kong University of Science and Technology, Hong Kong\\Email:\{eeknlau, eechungha\}@ust.hk}}

\maketitle
\begin{abstract}
\footnotetext{This work has been supported by Huawei Technologies.}
In this paper, we analyze the fundamental power-delay tradeoff in point-to-point OFDM systems under imperfect channel state information quality and non-ideal circuit power. We consider the dynamic backpressure (DBP) algorithm, where the transmitter determines the rate and power control actions based on the instantaneous channel state information (CSIT) and the queue state information (QSI). We exploit a general fluid queue dynamics using a continuous time dynamic equation. Using the sample-path approach and renewal theory, we decompose the average delay in terms of multiple {\em unfinished works} along a sample path, and derive an upper bound on the average delay under the DBP power control, which is asymptotically accurate at small delay regime. We show that despite imperfect CSIT quality and non-ideal circuit power, the average power $(P)$ of the DBP policy scales with delay $(D)$ as $P=\mathcal O(D \exp(1/D))$ at small delay regime. {\red While the impacts of CSIT quality and circuit power appears as the coefficients of the scaling law, they may be significant in some operating regimes. }
\end{abstract}

\newtheorem{Result}{Result}
\newtheorem{Problem}{Problem}
\newtheorem{Definition}{Definition}
\newtheorem{Lemma}{Lemma}
\newtheorem{Assumption}{Assumption}
\newtheorem{Remark}{Remark}
\newtheorem{Theorem}{Theorem}
\newtheorem{Corollary}{Corollary}


\section{Introduction}

There is a growing awareness of energy efficiency of the wireless infrastructure. In \cite{meshkati_energy-efficient_2007, mingbo_xiao_utility-based_2003}, the authors considered adaptive power control to optimize an energy efficiency metric, namely the Joule per bit. This metric measures the average cost (i.e. energy) to the generated utility (i.e. information bits) and allows for accessing the energy efficiency at full loads. However, this metric does not incorporate the delay aspect. Furthermore, in all these works, perfect channel state information (CSIT) is assumed and the issues of packet errors are ignored.
There have been several works that consider robust power control by taking into account of the imperfect CSIT. In \cite{rey_robust_2005, yingwei_yao_rate-maximizing_2005}, the authors derived a power adaptation algorithm to maximize the system goodput accounting for potential packet errors due to imperfect CSIT.
There are also some works that consider MIMO design with imperfect CSIT \cite{rey_robust_2005, yoo_capacity_2006} or limited feedback \cite{love_overview_2008}. While these works have dealt with the impact of imperfect CSIT, they have ignored the burstiness as well as the delay performance of the source data.

In general, it is quite challenging to consider delay optimization in wireless systems because that involves a joint consideration of both the information theory (to model the PHY dynamics) and the queueing theory (to model the delay dynamics).
One general approach to delay optimal control is to use Markov Decision Process (MDP) and the optimal control is given by the Bellman equation \cite{bertsekas_dynamic_1987}. However, it is well known that there is no simple solution to the Bellman equation and one cannot obtain viable solutions using brute force value iteration or policy iteration.
Another approach is to adopt Lyapunov theory to derive the {\em throughput optimal} control policy\footnote{Throughput-optimal policies  are  the  set of policies which  can  stabilize the queues in the system if the arrival rate vector is inside the \emph{stability region} \cite{tassiulas_stability_1992}.} \cite{georgiadis_resource_2006}.
While the control policy derived is adaptive to both the CSIT and QSI, the delay performance of such schemes are not fully understood. In {\blue\cite{neely_energy_2006},} the authors derived the {\em dynamic backpressure} (DBP) control
and the well-known asymptotic power-delay tradeoff of $\mathcal{O}(1/V)$, $\mathcal{O}(V)$ {\red as $V\to\infty$, i.e., at large delay regime\footnote{\red $V$ is a
constant which determines the tradeoff between power consumption and delay performance. We use the following notation in \cite{berry_optimal_2006} to characterize the asymptotic behavior a function $g(x)$ as $x\to x^*$: $g(x)=\mathcal{O}(f(x))$, if ${\lim\sup}_{x\to x^*}\frac{g(x)}{f(x)}<\infty$; $g(x)=\Omega(f(x))$, if ${\lim\sup}_{x\to x^*}\frac{f(x)}{g(x)}<\infty$; $g(x)=\Theta(f(x))$, if $g(x)=\mathcal{O}(f(x))$, and $g(x)=\Omega(f(x))$.}.} Furthermore, all these existing works have assumed perfect CSIT and it is not clear how the CSIT quality will affect the underlying delay power tradeoff.

In this paper, we are interested to study the inter-relationship between energy efficiency, CSIT quality and delay performance of the DBP algorithm for OFDM systems. The following are some first order technical challenges that have to be overcome:
\begin{itemize}

\item \textbf{Coupling between the Control Policy and the Queue Dynamics:} The DBP algorithm is adaptive to both the QSI and CSIT and hence, this introduces coupling to the queue dynamics and the control actions at each frame. This coupling makes the delay analysis of the system extremely difficult because there is no closed-form expression for the steady state distribution of the queue length for such dynamic policy. While there are some works in the literature that analyze the delay-power tradeoff for DBP algorithms using Lyapunov bounds \cite{neely_energy_2006}, such analyses are focused on the first order analysis. However, such Lyapunov bounding techniques are usually very loose and cannot be used to study {\red the second order impacts of system parameters} such as how the CSIT quality affects the delay performance.

\item \textbf{Coupling of Imperfect CSIT and Energy Efficiency:} For a given delay requirement, we should reduce the transmission data rate so as to increase the transmission time \cite{qiao_miser:_2003, zafer_minimum_2009} for energy efficient communications\footnote{This is because the transmission power increases exponentially with data rate.}. However, {\red the underlying tradeoff changes} when circuit power and imperfect CSIT are taken into account. By circuit power, we refer to the power consumption in the RF transmission chain, which can be assumed to be constant irrespective of the transmission rate. Thus, it is not always energy efficient to lengthen the transmission time up to the delay limit because the energy expenditures resulting from circuit power are proportional to the transmission time. On the other hand, due to imperfect CSIT, the transmission power required to support a certain goodput depends on the CSIT quality. Yet, quantifying this relationship for OFDM systems with joint encoding across subcarriers is not trivial because that involves finding the CDF of the mutual information.

\item \textbf{Performance in Small Delay Regime:} Another potential limitation of the Lyapunov bounding technique is that it works for asymptotically large delay regime, which is usually not the regime we are interested in. There are not many works on the analysis of small delay regions. In \cite{berry_optimal_2006}, the authors studied the delay-power tradeoff of a bounded rate scheme in fading channels at small delay regime under perfect CSIT. However, the delay performance of DBP algorithm in OFDM systems at small delay regime as well as the impact of CSIT quality are still not well-understood.
\end{itemize}

In this paper, we overcome the above challenges using continuous time stochastic calculus \cite{kushner_numerical_2001, bertsekas_dynamic_2007}. For instance, we introduce a {\em virtual continuous time system} (VCTS) and model the fluid queue dynamics using a continuous time dynamic equation \cite{zafer_calculus_2005, zafer_minimum_2009}. Using a calculus approach, we first derive closed-form expressions for the {\em unfinished work} in between arrivals. To analyze the end-to-end average delay performance, we adopt and extend the sample-path approach \cite{bertsekas_dynamic_2007} and decompose the average delay in terms of multiple {\em \red unfinished works} along a sample path using renewal reward theory. As such, we obtained an asymptotically accurate upper bound on the average delay under the DBP power control at small delay regime. We show that despite imperfect CSIT quality and non-ideal circuit power, the average power $(P)$ of the DBP policy scales with delay $(D)$ as $P=\mathcal O(D \exp(1/D))$ at small delay regime. {\red While the impacts of CSIT quality and circuit power appears as the coefficients of the scaling law, they may be significant in some operating regimes.}
\begin{table}[b!]
\renewcommand{\arraystretch}{0.6}
\caption{Notation used throughout the paper}
\label{tab_not}
\centering
\begin{tabular}{|c|l|} \hline
{Symbol}& {Meaning}   \\

\hline
$\Delta t$  & time duration of a scheduling slot\\
$T$  & time duration of a arrival period\\
$\overline B$  & the average of arriving packet size\\
$U(k)$   & total number of remaining bits at $k$-th scheduling slot in a queue\\
$\widetilde U(t)$   & fictitious queue state at continuous time $t$\\
$\varepsilon$  & target packet error rate \\
$n_F$  & the number of subcarriers\\
$\sigma^2_e$  & the variance of the CSIT error \\
$\overline {{j^U}}$   & per-period average unfinished work   \\
 $\overline {{j^g}}$  & per-period average energy consumption  \\
 $\overline {{J^U}}$  & continuous time per-period average unfinished work \\
 $\overline {{J^g}}$  & continuous time per-period average energy consumption \\
 $\bar d(\Omega)$  & average (end-to-end) delay of a policy $\Omega$\\
 $\bar g(\Omega)$  & average power consumption of a policy $\Omega$ \\
\hline
\end{tabular}
\end{table}

\section{System Model}

In this section, we shall elaborate the system model of the OFDM link, including the physical layer model, the bursty source model, the queueing dynamics and the power consumption model.

\subsection{Frequency Selective Fading Channel Model and the Imperfect CSIT Model}

We consider a point-to-point OFDM system with $n_F$ subcarriers. The number of resolvable paths in the frequency selective channel is given by $N_d = \left\lfloor {\frac{W}{{\Delta {f_c}}}} \right\rfloor$, where $W$ is the signal bandwidth and $\Delta {f_c}$ is the coherence bandwidth. The channel impulse response can be described by:
\[h(t,\upsilon ) = \sum\nolimits_{l = 1}^{{N_d}} {{h_l}(t)\,} \delta \left( {\upsilon  - \frac{l}{W}} \right)\]
where $\frac{l}{W}$ is the time delay of the $l$-th path and $h_l(t)$ is the corresponding circularly symmetric complex Gaussian (CSCG) random fading coefficients with zero mean and {\red variance $\sigma_l^2$ ($\sigma_l^2$ defines the power-delay profile).} Using $n_F$-point IFFT and FFT in the OFDM system, the received signal in the frequency domain is given by:
\[{Z_n} = \,\,{H_n} \cdot {S_n}\,\, + \,\,{w_n}, \]
where $S_n$ and $Z_n$ are the transmit and receive signals, respectively, of the $n$-th subcarrier and $w_n$ is the i.i.d. complex Gaussian noise with zero mean and normalized variance $1/n_F$ (so that the total noise power across the $n_F$ subcarriers is unity).
Note that $H_n\,\, = \,\,\sum\nolimits_{l = 1}^{N_d} {{h_l}\,{e^{\frac{{ - j2\pi ln}}{{{n_F}}}}}\,}$ for all $n$, which is the FFT of the time-domain channel fading coefficients $\{h_0, \cdots, h_{N_d-1}\}$.


For simplicity, we consider a TDD system and the transmitter obtains an estimate of the CSIT based on the reciprocal reverse channel \cite{marzetta_fast_2006}. However, due to the channel estimation noise as well as the TDD duplexing delay, the estimated CSIT may be outdated. Assume that the CSIT is estimated using MMSE prediction in the time domain, the CSIT model in the time domain is given by:
\[{\widehat h_l}\,\, = \,\,{h_l}\,\,\, + \,\,\Delta {h_l},\,\,\,\,\,{\red \Delta {h_l} \sim \mathcal{CN}\big(0,\sigma_{h,l}^2\big),\,\,\,l\,\, \in \,\{ 0,1,\,\,...\,,N_d - 1\}} .\]
{\red  where $\sigma_{h,l}^2=1-\frac{E_p\sigma_l^2}{E_p\sigma_l^2+1}J_0(2\pi f_D\tau)$, $E_p$ is the pilot SNR, $J_0$ is a Bessel function of the first kind of order 0, $f_D$ is the Doppler shift, and $\tau$ is the duplexing delay\cite{csi_error_it}},{\blue \cite{Ramyadelayfb2009}}. Thus, the estimated CSIT in the frequency domain $\widehat H_n$ after $n_F$-point FFT of $\{ {\widehat h_1}\,\,,...,{\widehat h_{N_d - 1}}\,\}$ is as follows:
\begin{equation}\label{eq_H}
{\widehat H_n}\,\, = \,\,{H_n}\,\,\, + \,\,\Delta {H_n}
\end{equation}
where $H_n$ is the actual channel state information (CSI) of the $n$-th subcarrier and $\Delta {H_n}$ represents the CSIT error. \red The CSIT errors $\Delta {H_n}$ is CSCG with zero mean and variance $\sigma_n^2=\sum\nolimits_{l=1}^{N_d}\sigma_{h,l}^2$, and the correlation of the CSIT error between the $n_1$-th and $n_2$-th subcarriers is given by:
$\mathbb{E}[\Delta {H_{{n_1}}}\Delta {H_{{n_2}}}^H] = \sum\nolimits_{l=1}^{N_d}\sigma_{h,l}^2e^{\frac{ - j2\pi l(n_1-n_2)}{n_F}}$.

\black
We model the packet error solely by the probability that the scheduled data rate exceeds the instantaneous mutual information. Note that the packet errors due to imperfect CSIT is systematic and cannot be eliminated by simply using strong channel coding. Therefore we shall exploit diversity to protect the information from channel outage to enhance the chance of successful delivery to the receiver. Specially, the encoded symbols are transmitted over the frequency domain via a random frequency interleaver. The conditional packet error probability (PER) $\varepsilon$ (conditioned on the CSIT ${\widehat{\bf{H}}}$ and the QSI $U$) of a transmission with data rate $r$ (nat/sec) is given by: $\,\Pr \left[ {r\,\, > \,\sum\nolimits_{n = 0}^{{n_F} - 1} {\log \left( {1 + \frac{{{P_{tx}}{{\left| {{H_n}} \right|}^2}}}{{{n_F}}}} \right)} \left| {\widehat{\bf{H}}} \right.\,,\,\,U} \right]$ where $P_{tx}$ is the transmit power. $\widehat{\bf{H}}\,\, \buildrel \Delta \over = \,\,{\left( {{{\widehat H}_0}\,,\,...\,,{{\widehat{H}}_{{n_F} - 1}}} \right)^T}$ indicates the estimated CSIT which is described in (\ref{eq_H}) and $U$ denotes the queue length (in number of bits) [$U$ is defined in Section II.B].

{\red
\begin{Remark}[Power Allocation over Subcarriers]
In this paper, we have assumed uniform power allocation over subcarriers for the following reasons. First, there is no known closed form expression conditional cdf of mutual information $\sum\nolimits_{n = 0}^{{n_F} - 1} {\log \left( {1 + \frac{{{P_{tx}}{{\left| {{H_n}} \right|}^2}}}{{{n_F}}}} \right)}$. Second, for moderate CSIT quality, the performance bottleneck is the packet errors (due to the residual uncertainty of mutual information given inaccurate CSIT). As a result the first order factor is diversity. Based on uniform power allocation, we have shown (Lemma \ref{lem_csit}) that full diversity order can be captured. Third, we have compared the performance of optimized power allocation  {\blue (obtained by stochastic gradient method \cite{BertsekasNeuro:1996,Borkarbook:2008})} versus uniform power allocation in Fig. \ref{fig_gradient} and it is shown that power optimization over subcarriers only shows marginal gain at moderate CSIT quality. ~\hfill\IEEEQED
\end{Remark}}

%

\subsection{Bursty Source Model, Queue Dynamics and Dynamic Rate Control Policy}
In this paper, we consider a bursty bit flow model. The time dimension is partitioned into scheduling slots of duration $\triangle t$ and indexed by $k$. Multiple slots are grouped as a frame of duration $T$ and indexed by $m$ as illustrated in Fig. \ref{fig_slotmodel}. We have the following assumptions regarding the CSI and the bursty source model.

\begin{Assumption}[Quasi-Static CSI]
 For notation convenience, we denote $\widehat{\bf{H}}(k)\,\, \buildrel \Delta \over = \,\,{\big( {{{\widehat H}_0}(k)\,,\,...\,,{{\widehat{H}}_{{n_F} - 1}}(k)} \big)^T}$ and ${\bf{H}}(k)\,\, \buildrel \Delta \over = \,{\left( {{H_0}(k)\,,\,...\,,{H_{{n_F} - 1}}(k)} \right)^T}\,$ as the ${n_F}\, \times \,\,1$ dimension CSIT and CSI vectors, respectively, at the $k$-th scheduling slot. The CSI ${\bf{H}}(k)$ is assumed to be quasi-static within a scheduling slot and i.i.d. between scheduling slots. ~\hfill\IEEEQED
\end{Assumption}

\begin{Assumption}[Bursty Source Model]
 Let $\mathcal B_m$ be the random new arrivals (in bits) at the $m$-th frame. The arrival process $\{\mathcal B_m\}$ is i.i.d. over $m$ according to a general distribution $\Pr (\,\mathcal B\,)$ with average $\mathbb{E}[{\mathcal B_m}]=T\overline{B}$ where $\overline{B}$ is the average arrival rate per second. ~\hfill\IEEEQED
\end{Assumption}

\begin{figure}[t!]
\centering
\includegraphics[width=.6\columnwidth]{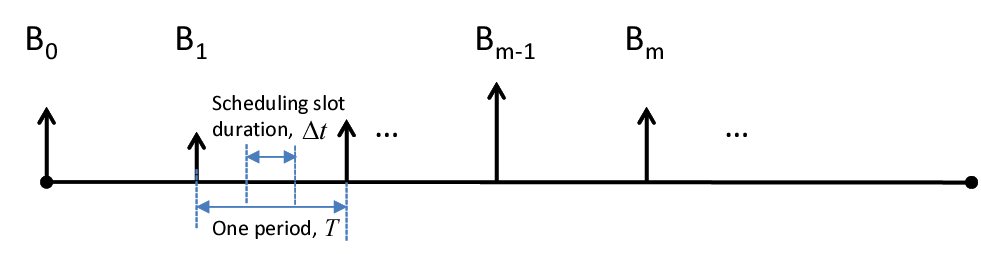}
\caption{The random arrival and delay deadline model.} \label{fig_slotmodel}
\end{figure}

The QSI $U(k)$ is defined as the {\em unfinished work} (i.e. the total number of remaining bits) in the queue at the beginning of the $k$-th scheduling slot. Let $\chi (k)\,\, = \,\,\big( {U(k),\,\,\widehat{\bf{H}}(k)\,} \big)$ be the system state at the $k$-th scheduling slot. Given an observed system state $\chi$, the transmitter adjusts the transmit data rate according to a stationary rate control policy defined below.

\begin{Definition}[Stationary Rate Control Policy]
Let $r(k)$ be the rate allocation action of the OFDM transmitter at the $k$-th scheduling slot. A stationary rate control policy $\Omega$ is a mapping from the system state $\chi$ to a rate control action $r$. Specifically, $r(k) = \Omega(\chi(k))$ for all $k$. ~\hfill\IEEEQED
\end{Definition}

Given a stationary rate control policy $\Omega$, the queue dynamics is given by:
\begin{equation}\label{eq_Ut}
U(k + 1)\,\, = \,\,{\left[ {U(k)\,\, - r(k)(1 - {e(k)})\Delta t\,} \right]^ + } + \,\,{\mathcal B_{\left\lfloor {\frac{{k\Delta t}}{T}} \right\rfloor }}{\bf{1}}\big( {\bmod ( k,T/\Delta t ) = 0} \big)
\end{equation}
where ${x^ + }\,\, = \,\,\max \,\{ x,\,\,0\} $ and $e(k) \in\{0,1\}$ is the packet error indicator at the $k$-th scheduling slot\footnote{We assume there is an error-free and delay-free ACK/NAK feedback from the receiver to the transmitter.}.

\subsection{Power Consumption Model}

At the transmitter, the power consumption is contributed by the {\em transmission power} of the power amplifier and the {\em circuit power} of the RF chains (such as the mixers, synthesizers, phase-lock loop and digital-to-analog converters). The transmission power $P_{tx}$ in general depends on the transmitted data rate $r$ as well as the CSIT quality. {\red Specifically, the transmitted data rate is given by $r = n_F\log(1+\frac{P}{n_F}f(\varepsilon,\sigma_e^2,\mathbf{\hat{H}}) )$ for a target PER $\varepsilon$ \cite{lau_asymptotic_2008}, where $f(\varepsilon,\sigma_e^2,\mathbf{\hat{H}})$ is a ``black box function'' which characterize the behavior of the underlying PHY under imperfect CSIT. In other words, for a given data rate $r$, CSIT error $\sigma^2_e$ and target PER $\varepsilon$, the minimum required transmission power is given by
\begin{equation}\label{eq_ptx}
{P_{tx}}(r;\,\,\widehat{\bf{H}})\,\, = \,\,\frac{{\left( {{e^{\frac{r}{{{n_F}}}}} - 1} \right){n_F}}}{{f(\varepsilon,\sigma_e^2,\mathbf{\hat{H}})}}
\end{equation}
For uniform power-delay profile\footnote{\red Uniform power delay profile is known to be the worst case profile in frequency selective fading channels \cite{wimax:2008,Proakis:2001}. As a result, the closed form expression for $f(\varepsilon,\sigma_e^2,\mathbf{\hat{H}})$ in Lemma \ref{lem_csit} represents the worst case profile. For general power delay profile, there is no closed form expression for $f(\varepsilon,\sigma_e^2,\mathbf{\hat{H}})$ but it can be obtained via offline PHY level simulation.}, the following lemma summarizes an asymptotically accurate relationship at high and low SNR.}
\begin{Lemma}[{\red Relationship between Transmit Power and CSIT Quality}]\label{lem_csit}
{\red Under uniform power-delay profile, for a given data rate $r$, CSIT error $\sigma^2_e$ and target PER $\varepsilon$, $f(\varepsilon,\sigma_e^2,\mathbf{\hat{H}})$ in \eqref{eq_ptx} is given by:}
\begin{equation}
{\red f(\varepsilon,\sigma_e^2,\mathbf{\hat{H}})\,\, \doteq \,\,F_{{\psi ^2};{s^2}}^{ - 1}\left( {{\varepsilon}} \right) }
\end{equation}
where $\doteq$ denotes asymptotic equality for high and low SNR and $F_{\psi^2;s^2}^{-1}$ is the inverse CDF of the non-central chi-square random variable $\psi^2$ with non-centrality parameter $s^2$. The chi-square random variable ${\psi ^2} = \frac{1}{{{N_d}}}\sum\nolimits_{n \in {I_B}} {{{\left| {{H_n}} \right|}^2}} $  has $2N_d$ degrees of freedom and variance $\sigma^2_e/N_d$, where $I_B$ is the set of $N_d$ independent subcarriers. The non-centrality parameter is given by ${s^2}({I_B}) = \frac{1}{{{N_d}}}\sum\nolimits_{n \in {I_B}} {|{{\widehat H}_n}{|^2}}$ . ~\hfill\IEEEQED
\end{Lemma}
\begin{proof}
Due to page limitation, please refer to \cite{lau_asymptotic_2008} for the proof.
\end{proof}

\black

The transmission power $P_{tx}(r;\,\,\widehat{\bf{H}})$ is a convex increasing function of the data rate $r$. On the other hand, the CSIT quality affects the $P_{tx}(r;\,\,\widehat{\bf{H}})$ via the variance of the inverse chi-square CDF $F_{{\psi ^2};{s^2}}^{ - 1}(x)$. Specifically, a larger transmission power is required for the same data rate $r$ when the CSIT error increases. Fig. \ref{fig2} illustrates the $P_{tx}(r;\,\,\widehat{\bf{H}})$ versus $r$ at different CSIT errors {\red under uniform power-delay profile}. Observe that a larger transmission power is required for the same data rate $r$ when the CSIT error increases.

\begin{figure}[t!]
\centering
\includegraphics[width=.5\columnwidth]{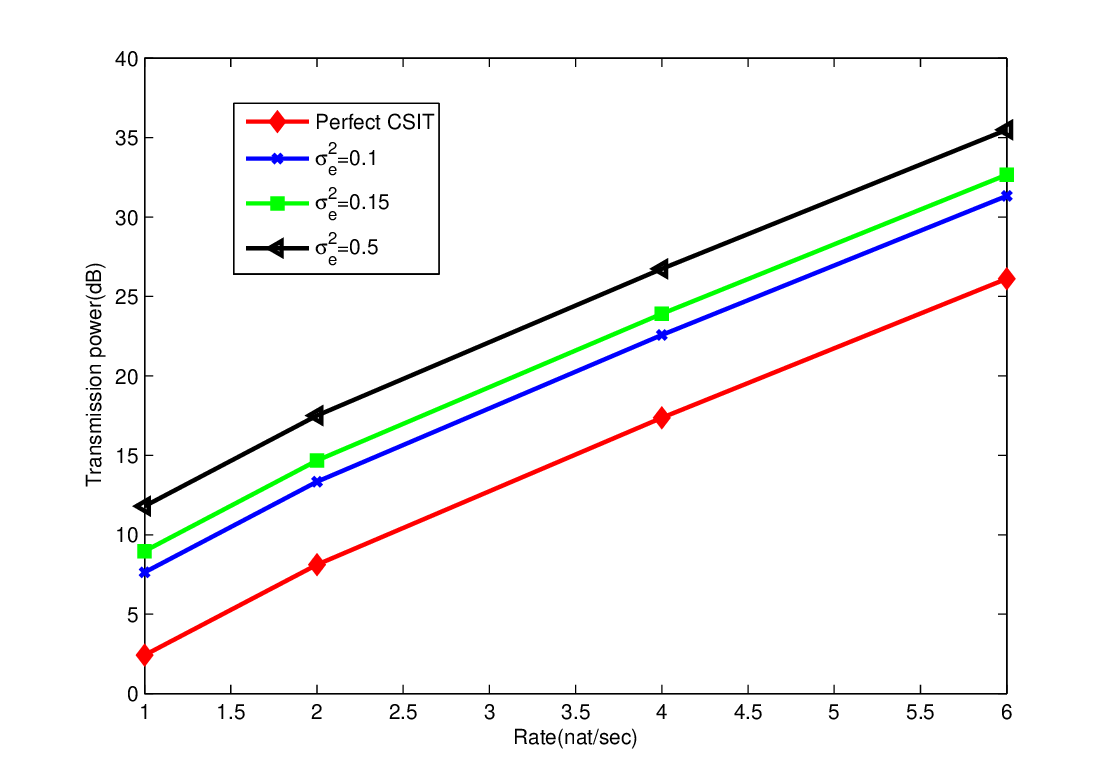}
\caption{Required transmission power (dB) versus data rate with different CSIT error $\sigma^2_e$, target PER=0.01, single subcarrier, bandwidth = 1Hz.} \label{fig2}
\end{figure}

On the other hand, the circuit power $P_{cct}$ is constant regardless of the transmission data rate. The transmission circuit is assumed to be on during the transmission of a burst. Hence, the total power consumption at the transmitter during a burst is given by:
\[g\left( {r,\,\,\,\chi } \right)\,\, = {P_{tx}}\left( {r,\,\widehat{\bf{H}}} \right)\, + {P_{cct}}. \]

Note that there are two conflicting dynamics due to (a) bursty arrivals and (b) energy efficiency regarding whether one should increase the transmission time to finish the data backlog. It is known \cite{berry_communication_2002} that one should try to reduce the data rate when the buffer level is small so as to decrease the chance of the buffer being empty.
Furthermore, the transmission power increases exponentially with the data rate and hence, one should increase the transmission time so as to save energy \cite{zafer_minimum_2009}. On the other hand, when the circuit power is taken into consideration, the tradeoff dynamics will be changed because the energy expenditures resulting from circuit power are proportional to the transmission time of the system\cite{shuguang_cui_energy-constrained_2005}. As a result, it is not always advantageous to increase the transmission time.

\section{Dynamic BackPressure Power Control}

In this section, we first focus on a queue stabilization problem and derive the throughput optimal policy DBP using the {\red Lyapunov function $L(U) = U^{2}/{{2 }}$ \cite{andrews_scheduling_2004}}.

\subsection{Preliminaries of Stochastic Stability}
We first introduce a few definitions. We say that the OFDM link with bursty arrival is \emph{strongly stable} if: \[\lim \mathop {\sup }\limits_{K \to \infty } \frac{1}{K}\sum\nolimits_{k = 1}^K {\mathbb{E}[U(k)]} \,\, < \,\,\infty \,\,.\]

\begin{Definition}[Stability Region] The stability region $\Lambda_{\Omega}$ of policy $\Omega$ is the set of average arrival rates $\lambda$ for which the system is stable under $\Omega$. The stability region of the system $\Lambda$ is the closure of the set of all average arrival rates $\lambda$ for which a stabilizing control policy exists. Mathematically, we have $\Lambda \,\, = \,\,\mathop { \cup \,}\limits_{\Omega  \in G} \,{\Lambda _\Omega }$, where $G$ denotes policy space. ~\hfill\IEEEQED
\end{Definition}

\begin{Definition}[Throughput-Optimal Policy] A throughput-optimal policy dominates\footnote{A policy $\Omega_1$ \emph{dominates} another policy $\Omega_2$ if $\Lambda_{\Omega_2} \subset \Lambda_{\Omega_1}$.} any other policy in $G$, i.e. it has a stability region that is the superset of the stability region of any other policy in $G$. Therefore, it should have a stability region equal to $\Lambda$. ~\hfill\IEEEQED
\end{Definition}

In  other  words,  throughput-optimal  policies  ensure  that  the  queueing  system  is  stable  as long as the vector arrival rate is within the system stability region $\Lambda$. Note that the throughput optimal policy is not unique. While throughput optimality does not guarantee delay optimality in the system, the former policy can still improve the delay performance. Furthermore, using Lyapunov analysis techniques \cite{georgiadis_resource_2006}, the throughput optimal policy derived usually has a simple form, which is desirable for implementation.


Consider the Lyapunov function $L(U) = {U^{2}{{}}}/{{2 }}$ \cite{andrews_scheduling_2004},
and define the one-step Lypunov drift $\Delta L(U)$ as:
\begin{equation}\label{Lyadef}
\Delta \left( {L(U(k) )} \right) \buildrel \Delta \over = \mathbb E \left[ {L(U(k + 1) ) - L(U(k) )\left| {\,\,U(k)} \right.} \right].
\end{equation}

The following Lemma summarizes the results on the Lyapunov drift. {\red
\begin{Lemma}[Lyapunov Drift]\label{lem_lyaD} Let $A(k)$ be the arrival process of queue dynamics (\ref{eq_Ut}). Furthermore, let $A(k)\, \le \,{A_{\max }}$ and $r(k)(1 - e(k) )\Delta t \le \,{R_{\max }}$, for some positive $A_{max}$ and $R_{max}$. The one step Lyapunov drift for the OFDM link with imperfect CSIT is given by:
\begin{equation}\label{LyaD}
\Delta \left( {L(U(k))} \right)\,\, \le \,\,\,\frac{1}{2 }\big(A_{\max }^{2 } + R_{\max }^{2 }\big) - \mathbb E\left[ U{{(k)} }\left\{ {r(k)(1 - e(k) )\Delta t - A(k)} \right\}\,\left| {\,\,U(k)} \right. \right]
\end{equation}
\end{Lemma}
\begin{proof}
The proof follows similar technique as in \cite{georgiadis_resource_2006} and is omitted.
\end{proof}
}

\subsection{Dynamic BackPressure Algorithm}
Based on the Lyapunov drift $\Delta \left( {L(U(k) )} \right)$ in Lemma \ref{lem_lyaD}, the DBP algorithm can be derived by maximizing the negative drift term in \eqref{LyaD}. Given an observed state $\chi$ at any scheduling time slot, the instantaneous data rate $r_{DBP}(U,\widehat{\bf{H}})$ is given by:
\begin{equation}\label{eq_ropt1}
{r_{DBP}}\left( {U,\,\widehat{\bf{H}}} \right) = \mathop {\arg \max }\limits_r \,\,\left\{ {{U }r(1 - {\varepsilon})\Delta t\,\, - \,V\left( {{P_{tx}}(r;\widehat{\bf{H}}) + {P_{cct}}} \right)\Delta t} \right\}
\end{equation}
where $V$ is a constant which determines the tradeoff between power consumption and delay performance. Solving the problem in (\ref{eq_ropt1}) for the DBP, the instantaneous data rate ${r_{DBP}}(\chi)$ is given by:
\begin{equation}\label{eq_ropt2}
{r_{DBP}}\left( \chi  \right) = {n_F}{\bigg[ {\log \Big\{ {\frac{{{U }(1 - {\varepsilon})f(\epsilon,\sigma_e^2,\mathbf{\hat{H}})}}{V}} \Big\}} \bigg]^ + }.
\end{equation}

During a burst transmission, the instantaneous power consumption under DBP is given by:
\begin{equation}\label{eq_Popt}
g_{DBP}^{}\left( \chi  \right)\,\, = \,\left\{ {\begin{array}{*{20}{c}}
   {\frac{{{U }(1 - {\varepsilon}){n_F}}}{V}\, - \frac{{{n_F}}}{{f(\epsilon,\sigma_e^2,\mathbf{\hat{H}})}} + {P_{cct}},\,\,\,{\rm{if}}\,\,\frac{{{U}(1 - {\varepsilon}){n_F}}}{V}\, - \frac{{{n_F}}}{f(\epsilon,\sigma_e^2,\mathbf{\hat{H}})} > 0}  \\
   {0\,\,,\,\,\,\,\,\,\,\,\,\,\,\,\,\,\,\,\,\,\,\,\,\,\,\,\,\,\,\,\,\,\,\,\,\,\,\,\,\,\,\,\,\,\,\,\,\,\,{\rm{otherwise}.}\,\,\,\,\,\,\,\,\,\,\,}  \\
\end{array}} \right.
\end{equation}

\begin{Remark}[Multilevel Water-Filling Structure of the DBP]
The power control action in (\ref{eq_Popt}) is a function of both CSIT and QSI (where it depends on the CSIT indirectly via the noncentrality parameter ${s^2}({I_B}) = \frac{1}{{{N_d}}}\sum\nolimits_{n \in {I_B}} {|{{\widehat H}_n}{|^2}}$ in $f(\epsilon,\sigma_e^2,\mathbf{\hat{H}})$). It has the form of \emph{multilevel water-filling structure} where the transmission power is allocated according to the CSIT but the waterlevel is adaptive to the QSI. The parameter $V$ acts like the Lagrange Multiplier which determines the tradeoff between power consumption and delay. Furthermore, the $P_{cct}$ affects the power control (or rate control solution) in (\ref{eq_Popt}) and (\ref{eq_ropt2}) by introducing a penalty proportional to a burst transmission time. ~\hfill\IEEEQED
\end{Remark}


\section{Delay-Power Tradeoff of DBP with imperfect CSIT}

While the DBP is throughput optimal, we are interested in studying the end-to-end delay performance and the relationship between the average delay, average power and the CSIT quality.
In this section, we shall analyze the power-delay tradeoff using continuous time approximation and renewal process theory. {\red As we shall illustrate, this approach not only yields first order tradeoff relationship (at small delay regime) but also yields the second-order impacts due to imperfect CSIT and static circuit power $P_{cct}$.}

\subsection{Per-Period Unfinished Works}

To analyze the average delay, we first focus on the analysis of one arrival period $T$. Specifically, define the per-period average unfinished work $\overline{j^U}(U_0)$ and the per-period average energy consumption $\overline{j^g}(U_0)$ as:
\begin{equation}\label{eq_jUdef}
\overline {{j^U}} ({U_0})\, = \mathbb E\left[ {\sum\nolimits_{k = 0}^{N - 1} {{U_k}\Delta t\,\,\left| {{U_0}} \right.} } \right]
\end{equation}
\begin{equation}\label{eq_jgdef}
\overline {{j^g}} ({U_0})\, = \mathbb E\left[ {\sum\nolimits_{k = 0}^{N - 1} {g\left( {r(k)\,\,,\,\chi (k)} \right)\Delta t} \,\left| {{U_0}} \right.} \right]
\end{equation}
where $U_0$ is the leftover bits at the buffer at the starting epoch of a period $T$. To compute $\overline{j^U}(U_0)$ and $\overline{j^g}(U_0)$, we shall adopt a continuous time approach. Specifically, we define a {\em virtual continuous time system} (VCTS) as follows:
\begin{Definition}[Virtual Continuous Time Systems]
A virtual continuous time system is a fictitious system with a continuous queue state $\widetilde{U}(t)$ and channel state $\widetilde{\bf H}(t)$. The fictitious queue state evolves according to the following dynamic equation:
\begin{equation}\label{eq_ode}
\frac{{d\widetilde U(t)}}{{dt}}\,\, =  - \mathbb E\left[ {{r^*}\left( {\widetilde U(t),\,\,\widetilde {\bf H}(t)} \right)\left| {\widetilde U} \right.(t)} \right](1 - \varepsilon )\,
\end{equation}
where $r^{*}$ is the DBP rate policy in (\ref{eq_ropt2}) and the fictitious channel state is a white process\footnote{\red It means that for any given $t$, the distribution of random variable $\widetilde{\mathbf{H}}(t)$ is the same as that of the actual CSIT $\hat{\mathbf{H}}$. Furthermore, we have
$\mathbb{E}[\widetilde{\mathbf{H}}(t)]=\mathbb{E}[\hat{\mathbf{H}}]=\mathbf{0},
\mathbb{E}[\widetilde{\mathbf{H}}(t_1)\widetilde{\mathbf{H}}(t_2)]=\mathbb{E}[\hat{\mathbf{H}}\hat{\mathbf{H}}^H]\delta(t_1-t_2)$.} with identical distribution as the actual CSIT $\widehat{\bf H}$. Using the VCTS, the corresponding continuous time \emph{average per-period unfinished work} $\overline {{J^U}}(\widetilde U,t)$ and \emph{average per-period energy consumption} $\overline {{J^g}}(\widetilde U,t)$ are defined as follows:
\begin{equation}\label{def_JU}
\overline {{J^U}} (\widetilde U,t) \buildrel \Delta \over = \mathbb E\left[ {\int_t^T {\widetilde U(s)\,\,ds} \left| {\widetilde U} \right.(t)} \right]
\end{equation}
\begin{equation}\label{def_Jg}
\overline {{J^g}} (\widetilde U,t) \buildrel \Delta \over = \mathbb E\left[ {\int_t^T {g\left( {{r^*}\left( {\widetilde U(s),\,\widetilde{\bf{H}}(s)} \right)\,,\,\widetilde\chi (s)} \right)\,\,ds} \left| {\,\widetilde U(t)} \right.} \right]
\end{equation}
where $\widetilde\chi (t) = \,\,\big( {\widetilde U(t),\,\widetilde{\bf{H}}(t)} \big)$. ~\hfill\IEEEQED
\end{Definition}

{\red Note that \eqref{def_JU},\eqref{def_Jg} are the continuous-time counterparts of the discrete-time versions in (\eqref{eq_jUdef}),(\eqref{eq_jgdef}). They measure the total queue length (total area) of the {\em queue trajectory} and total energy consumption (total area) of the {\em power trajectory} during an inter-arrival interval, respectively. The queue trajectory and $\{\overline {{j^U}},\overline {{J^U}}\}$ are illustrated in Fig. \ref{fig_trajectory}. } \black To derive the actual per-period average unfinished work and energy consumption $\overline{j^U}$ and $\overline{j^g}$ in the discrete time, we first determine the continuous time counterparts using the following lemma.

\begin{figure}[t!]
\centering
\includegraphics[width=.8\columnwidth]{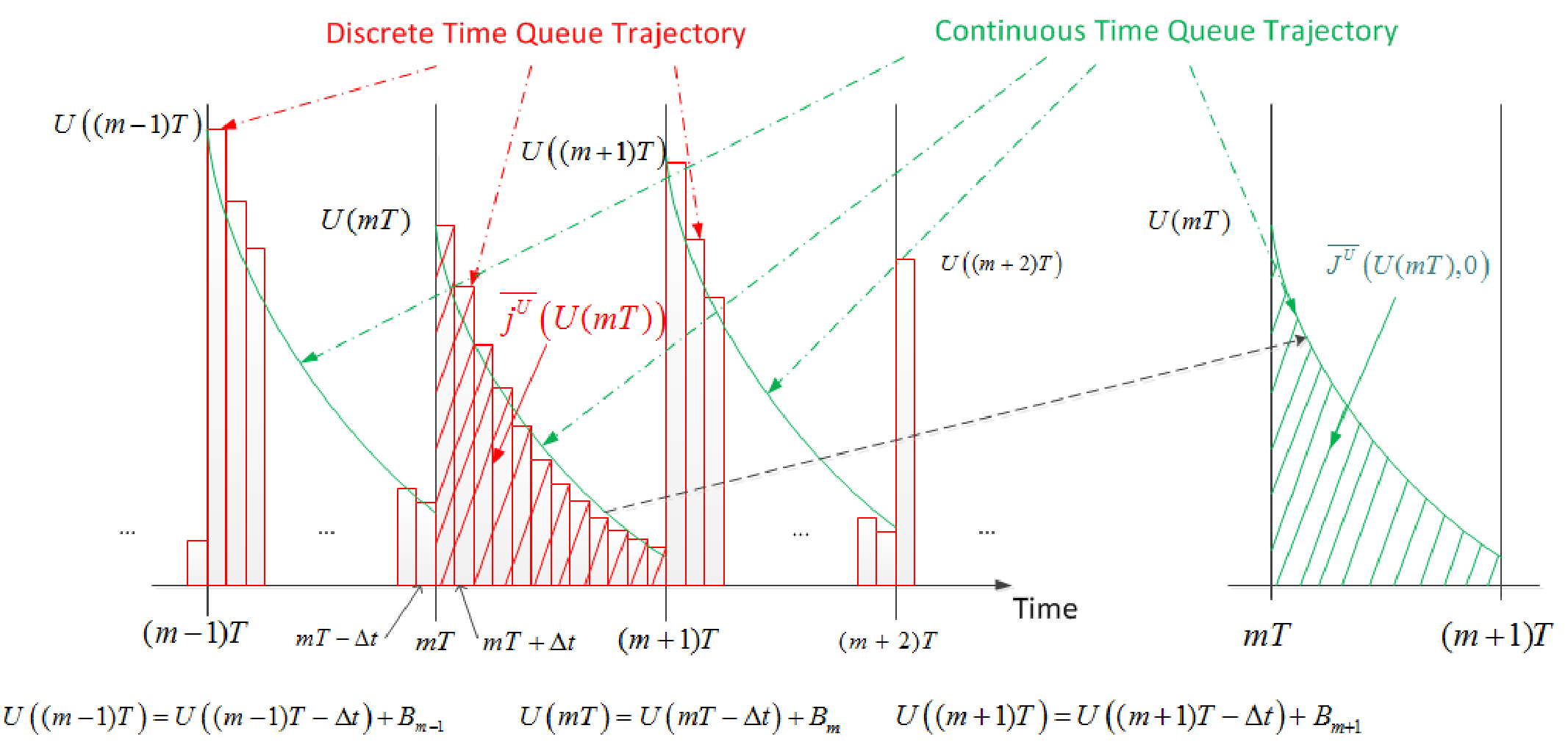}
\caption{{\red Illustration of the discrete time and VCTS queue trajectory and the unfinished works $\{\overline {{j^U}},\widetilde {{J^U}}\}$ in one inter-arrival interval $T$. } \black} \label{fig_trajectory}
\end{figure}

\begin{Lemma}[Verification Lemma of $\overline {{J^U}}$ and $\overline {{J^g}}$ {\red in VCTS}]\label{lem_hjb}
If there exists a continuous and differentiable function $\overline {J^U}^*(\widetilde U,t)$ and $\overline {J^g}^*(\widetilde U,t)$ satisfying the following partial differential equations (PDEs), respectively:
\begin{equation}\label{eq_hjb1}
- \frac{{\partial \overline {{J^U}} }}{{\partial \widetilde U}}(1 - \varepsilon )\mathbb E\left[ {{r^*}\left( {\widetilde U(t),\,\,\widetilde {\bf H}(t)} \right)\,\,\left| {\widetilde U(t)} \right.} \right] + \widetilde U(t) + \,\frac{{\partial \overline {{J^U}} }}{{\partial t}} = 0
\end{equation}
\begin{equation}\label{eq_hjb2}
\mathbb E\left[ {g({r^*}( \cdot ),\widetilde\chi (t))\,\left| {\widetilde U(t)} \right.} \right] - \frac{{\partial \overline {{J^g}} }}{{\partial \widetilde U}}(1 - \varepsilon )\mathbb E\left[ {{r^*}\left( {\widetilde U(t),\,\,\widetilde{\bf{H}}(t)} \right)\,\,\left| {\widetilde U(t)} \right.} \right] + \,\frac{{\partial \overline {{J^g}} }}{{\partial t}} = 0
\end{equation}
then, $\overline {J^U}^*(\widetilde U,t)$ and $\overline {J^g}^*(\widetilde U,t)$ are the total average unfinished work and total average energy consumption of the VCTS under the DBP in (\ref{eq_ode}). ~\hfill\IEEEQED
\end{Lemma}
\begin{proof}
{\red Note that (\ref{eq_hjb1}),(\ref{eq_hjb2}) resembles the Bellman equation of the discrete time dynamics in (\ref{eq_jUdef}), (\ref{eq_jgdef}). The proof is obtained using Taylor expansion of the value function and using the {\em divide-and-conquer} principle from (\ref{def_JU}),(\ref{def_Jg}). Please refer to Appendix A for details. Using Lemma \ref{lem_hjb}, the areas of the {\em queue trajectory} and the {\em power trajectory} can be obtained by solving the PDEs in (\ref{eq_hjb1}),(\ref{eq_hjb2}).}
\end{proof}

Finally, $\overline{j^U}(U_0)$ and $\overline{j^g}(U_0)$ are related to the continuous time counterparts by the following Theorem.
\begin{Theorem}[Relationship between the Continuous Time and Discrete Time Unfinished Works]\label{the_1}
 Let $\overline {J^U}^*(\widetilde U,t)$ and $\overline {J^g}^*(\widetilde U,t)$ be the solutions of the PDEs (\ref{eq_hjb1}) and (\ref{eq_hjb2}), respectively, in the virtual continuous time system. {\red For sufficiently small $\Delta t$,} the discrete-time per-period average unfinished work and energy consumption $\overline{j^{U}}$ and $\overline{j^{g}}$ are given by: ${\overline {j_{}^U}(U_0)}\,\, = \,\,\,{\overline {{J^U}} ^*(U_0,0)}\,\, + \mathcal O(\Delta t)$ and ${\overline {j_{}^g}(U_0)}\,\, = \,\,\,{\overline {{J^g}} ^*(U_0,0)}\,\, + \mathcal O(\Delta t)$. ~\hfill\IEEEQED
\end{Theorem}

Please refer to Appendix B for the proof. {\red As a result of Theorem \ref{the_1}, we can focus on the continuous time equations to solve the area of queue and power trajectories and we can be assured that the solutions obtained will be accurate up to $\mathcal O(\Delta t)$. } Using Lemma \ref{lem_hjb} and solving the associated PDEs in (\ref{eq_hjb1}) and (\ref{eq_hjb2}), we shall obtain an asymptotically accurate performance bounds of $\overline {J^U}^*(\widetilde U,t)$ and $\overline {J^g}^*(\widetilde U,t)$ which is summarized below.

\begin{Theorem}[Performance Bounds of $\overline {J^U}^*(\widetilde U,t)$ and $\overline {J^g}^*(\widetilde U,t)$ {\red in VCTS}]\label{the_bounds} The per-period average unfinished work $\overline {J^U}^*(\widetilde U,0)$ in (\ref{def_JU}) and energy consumption $\overline {J^g}^*(\widetilde U,0)$ in (\ref{def_Jg}) are given by
\begin{equation}\label{eq_JUup}
{\overline {{J^U}} ^*}(\widetilde U\,,0)\,\, \le \,\,\int_0^T {y(t;\,\beta )\,} dt\,
\end{equation}
\begin{equation}\label{eq_Jglow}
{\overline {{J^g}} ^*}(\widetilde U,0)\ge \int_0^T {{{\bigg[ {\frac{{{y }(t;\beta '){n_F}(1 - {\varepsilon})}}{V} + \mathbb E\Big[ {{{\Big[ {\frac{{\widetilde U_0 {n_F}(1 - {\varepsilon})}}{V} - \frac{{{n_F}}}{{f(\epsilon,\sigma_e^2,\mathbf{\hat{H}})}}} \Big]}^ + } + {P_{cct}}} \Big]- \frac{{\widetilde U_0 {n_F}(1 - {\varepsilon})}}{V}} \bigg]}^ + }dt}
\end{equation}
where $y(t;\,\beta ) = \exp \left[ { - \beta {\kern 1pt} {\kern 1pt}  + {\kern 1pt} {\rm{E}}{{\rm{i}}^{( - 1)}}\left[ {{\rm{Ei}}\left( {\log (\widetilde U_0^{}) + \beta} \right) - {n_F}(1 - {\varepsilon}){e^{\beta  }}t} \right]} \right]$,
$\beta  = \,\mathbb E\left[ {\log ((1 - {\varepsilon})\,f(\epsilon,\sigma_e^2,\mathbf{\hat{H}}))} \right]$,
and $\beta ' = \,\mathbb E\left[ \big(\log ((1 - {\varepsilon})\,f(\epsilon,\sigma_e^2,\mathbf{\hat{H}}))\big)^ + \right]$. ${\rm{Ei}}(x) = \,\int_{ - \infty }^x {{e^t}/t} \,dt$ for $x>0$, is the exponential integral function. The bounds are asymptotically accurate as V approaches 0. ~\hfill\IEEEQED
\end{Theorem}

Please refer to Appendix C for the proof.

\subsection{Delay Analysis for Deterministic Arrivals}

In this subsection, we establish the relationship between multiple per-period unfinished works and average end-to-end delay under the assumption of deterministic arrivals. Specifically, we assume the bit arrival $\mathcal B_m$ is deterministic (given by $B$). The average bit arrival rate (bits per seconds) is given by $\overline{B} = B/T$. Such an arrival model embraces VoIP as well as other delay-sensitive source models derived from {\em constant bit rate} (CBR) encoders \cite{eunkyung_kim_efficient_2007, narbutt_gauging_2006}. Let $\bar d(\Omega^*)$ and $\bar g(\Omega^*)$ be the average end-to-end delay and the average power consumption, respectively, under DBP policy $\Omega^*$ in actual discrete time systems. $\bar d(\Omega^*)$ and $\bar g(\Omega^*)$ are represented by the combination of per-period average unfinished work over multiple periods. {\red We first have the following results regarding the steady state leftover bits of the buffer in VCTS due to the accumulated arrivals in the previous arrival periods.
\begin{Lemma}[Steady State Leftover Bits for the VCTS]\label{lem_left} Let $L_m$ be the leftover bits of the buffer in VCTS at the end of the $m$-th arrival period. Given $L_{0}<L^*$, then we have $\sup_{m\in\mathbb{R}^+} L_m \leq L^*$, and $L^*>0$ satisfies the following fixed point equation:
\begin{equation}\label{eq_L}
{L^*}= {e^{ - \beta  + \text{Ei}^{( - 1)}\left[ {\text{Ei}\left( {\log (B + {L^*}) + \beta  } \right) -  (1 - {\varepsilon}){n_F}{e^{\beta  }}T} \right]}}
\end{equation}
Furthermore, the fixed point $L^*$ exists and is unique. ~\hfill\IEEEQED
\end{Lemma}
}

Please refer to Appendix D for the proof. Based on Lemma \ref{lem_left}, the average delay and power consumption for deterministic arrival is given by:

\begin{Theorem}[Average Delay and Average Power Consumption for Deterministic Arrivals]\label{the_dpbound} {\red For sufficiently small $\Delta t$,} the discrete-time average end-to-end delay $\bar d(\Omega^*)$ and average power consumption $\bar g(\Omega^*)$ under the DBP $\Omega^*$ with deterministic arrivals are given by:
\begin{equation}\label{eq_dup}
\overline d \left( {{\Omega ^*}} \right) \le \,\frac{1}{{\overline B }}\frac{1}{T}{\overline {{J^U}} ^*}(B + {L^*},0)\,\, + \,\mathcal O\left( {\Delta t} \right)
\end{equation}
\begin{equation}\label{eq_plow}
\overline g \left( {{\Omega ^*}} \right) \ge \frac{1}{T}{\overline {{J^g}} ^*}(B ,0)\,\, + \,\mathcal O\left({\Delta t} \right)
\end{equation}
where $B$ is the number of bits of an arrival packet in each period. ~\hfill\IEEEQED
\end{Theorem}
\begin{proof}
Please refer to Appendix E for the proof, {\red where the relationship between the discrete model and the VCTS in Theorem \ref{the_1} is utilized.}
\end{proof}
\begin{Corollary}[Asymptotic Power-Delay Tradeoff of DBP at Small Delay Regime]\label{cor_order} {\red For sufficiently small $V$}\footnote{\red $V$ is a parameter that determines the tradeoff between power and delay in the system. For a given data arrival rate, small $V$ corresponds to small delay regime.}, the asymptotic power-delay tradeoff of DBP  {\blue of the VCTS} is given by:
\begin{equation}\label{eq_dorde}
\overline d \left( {{\Omega ^*}} \right) =\mathcal O\left( {\frac{B^2}{{\log ({B }\mathbb{E}[f(\epsilon,\sigma_e^2,\mathbf{\hat{H}})]/V)}} + {\frac{V}{{\mathbb{E}[f(\epsilon,\sigma_e^2,\mathbf{\hat{H}})]}}}} \right)
\end{equation}
\begin{equation}\label{eq_gorde}
\overline g \left( {{\Omega ^*}} \right)= \Omega\left( {\left( {\frac{{{B }}}{V} + {P_{cct}}} \right)\frac{B}{{ \log ({B }\mathbb{E}[f(\epsilon,\sigma_e^2,\mathbf{\hat{H}})]/V)}}} \right)
\end{equation}
\end{Corollary}
\begin{proof}
Please refer to Appendix F for the proof.
\end{proof}

\begin{Remark}[Interpretation of Results] {\blue Note that from Theorem \ref{the_dpbound}, there is an additional $\mathcal O\left({\Delta t} \right)$ term in \eqref{eq_dup} and \eqref{eq_plow} accounting for the approximation error between the power/delay of the original discrete time system and the VCTS. Yet, to simplify discussion, we focus on the first-order comparisons from the power-delay tradeoff of the VCTS in Corollary \ref{cor_order}. }

\begin{itemize}
\item {\bf Comparison with CSIT-only policy:} The power-delay tradeoff result {\blue for the VCTS} in Corollary \ref{cor_order} is asymptotically accurate as $V\rightarrow 0$ and this corresponds to small delay regime. For a given CSIT quality and $P_{cct}$, the conditional average data rate {\red (conditioned on the queue state $\widetilde U(t)$)} for CSIT-only policy is given by:
\begin{equation}
{\overline r _{CSIT}} = \,E\left[ {{r_{CSIT}}(\widetilde{\bf{H}}(t))|\widetilde U(t)} \right] = \mathbb E\left[ {{n_F}{{\left\{ {\log \left( {\frac{{(1 - \varepsilon )f(\epsilon,\sigma_e^2,\mathbf{\widetilde{H}})}}{V}} \right)} \right\}}^ + }} \right].
\end{equation}
As a result, the power-delay tradeoff {\blue of the VCTS}  for CSIT-only policy \cite{berry_optimal_2006} at small delay regime\footnote{\red The delay expression for CSIT-only policy in \cite{berry_optimal_2006} 
{\blue is derived using a discrete time approach and is given by} $\overline{d}=\mathcal{O}(\frac{1}{\log(\overline{g})})+1$. {\blue On the other hand}, the result in this paper {\blue is  derived   using a continuous time approach (VCTS) and the delay is given by}  {\blue $\overline{d}=\mathcal O\left( \frac{1}{{\log (1/V)}} \right)+ \mathcal O(\Delta t)$ (for the actual discrete time system in terms of seconds).} Hence, they match each other when expressing in terms of seconds. } is given by $\overline g = \mathcal O(\exp(1/\overline d))$.  {\red On the other hand, since $\frac{1}{\log(1/V)}=\Omega(V)$, we have {\blue the delay and power of the VCTS given by} $
\overline d \left( {{\Omega ^*}} \right) = \mathcal O\left( \frac{1}{{\log (1/V)}} \right)$, and $
\overline g \left( {{\Omega ^*}} \right) = \Omega  \left( \frac{1 }{{ V \log (1/V)}} \right)$ from Corollary \ref{cor_order}.
Furthermore, since it is asymptotically accurate for small $V$ (small delay regime), we can conclude that the power-delay tradeoff for DBP {\blue in the VCTS} is given by $\overline g = \mathcal O(\overline d \exp(1/\overline d))$. Hence, compared with CSIT-only policy, the power consumption of DBP {\blue in the VCTS}  increases slower as delay $\overline d$ tends\footnote{{\blue Note that while the delay of the VCTS can go to zero as $V\to 0$, the delay of the actual discrete time system cannot go to zero and is given by $\mathcal O(\Delta t)$ when $V\to0$. This footnote applies to Corollary \ref{cor_orderr} as well.}}  to 0.} Furthermore, we could achieve this superior tradeoff performance even with imperfect CSIT quality and non-ideal circuit power $(P_{cct} > 0)$.

\item {\bf Effects of CSIT quality:} The penalty of CSIT quality is contained in {\red $\mathbb{E}[f(\epsilon,\sigma_e^2,\mathbf{\hat{H}})]$}, which appears in the coefficients of the tradeoff equations in (\ref{eq_dorde}) and (\ref{eq_gorde}). For a given target PER $\varepsilon$, a larger CSIT error corresponds to a smaller {\red $\mathbb{E}[f(\epsilon,\sigma_e^2,\mathbf{\hat{H}})]$}. Fig. \ref{fig_EF} illustrates $\mathbb{E}[F_{{\psi ^2};{s^2}}^{ - 1}( \varepsilon)]$ versus CSIT errors $\sigma_e^2$ for target PER $10^{-2}$, $10^{-3}$ and $10^{-4}$ {\red under the uniform power-delay profile}.

\item {\bf Effects of $P_{cct}$:} From (\ref{eq_gorde}), the average power consumption has two components, namely the transmission power and the circuit power. The term $P_{cct}/({{ \log ({B }f(\epsilon,\sigma_e^2,\mathbf{\hat{H}})/V)}}))$ corresponds to the circuit power consumption, which increases with the burst transmission time in one arrival period. For small $V$, the burst transmission time decreases in the order of $\mathcal O(1/({{ \log ({B }f(\epsilon,\sigma_e^2,\mathbf{\hat{H}})/V)}}))$.
~\hfill\IEEEQED
\end{itemize}
\end{Remark}

\begin{figure}[t!]
\centering
\includegraphics[width=.5\columnwidth]{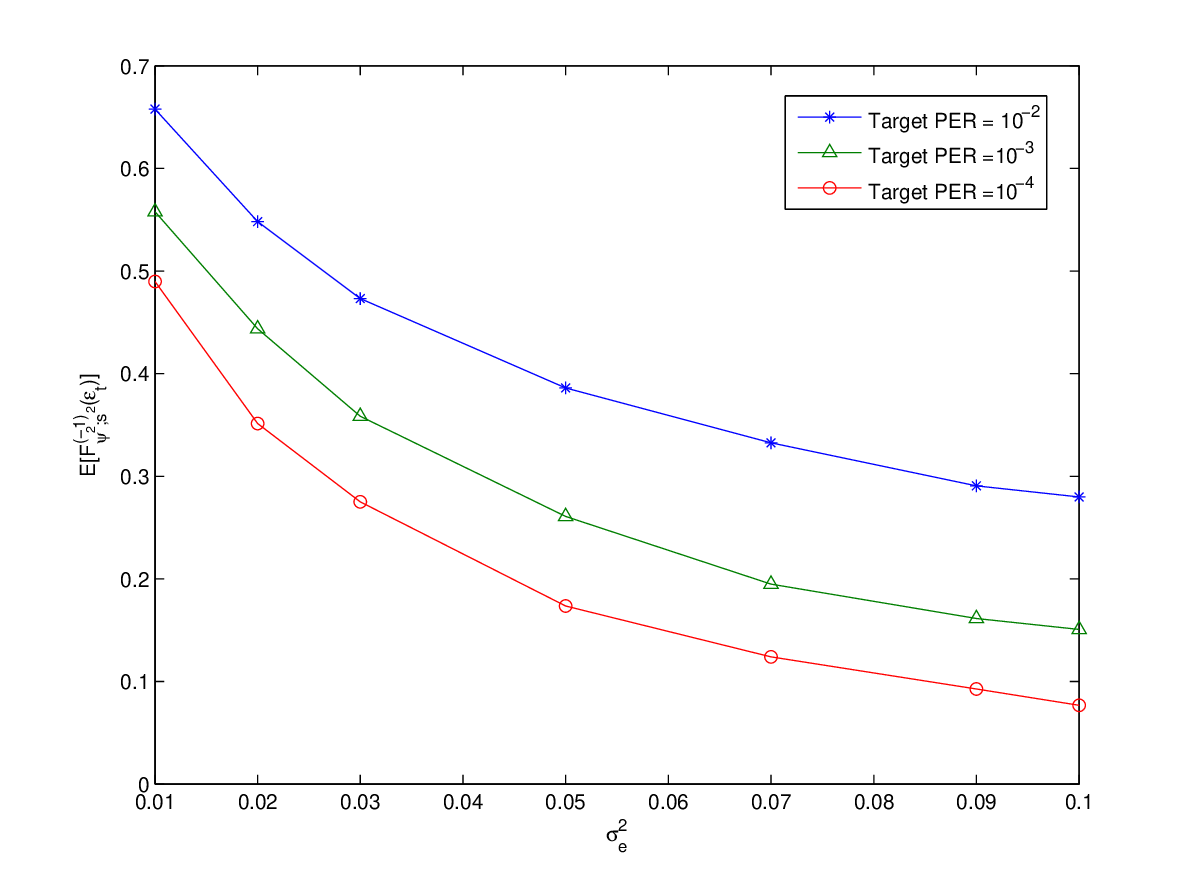}
\caption{Expectation of $F_{{\psi ^2};{s^2}}^{ - 1}( \varepsilon)$ versus CSIT error $\sigma^2_e$ at different target PER.} \label{fig_EF}
\end{figure}

\subsection{Delay Analysis for Random Arrivals}

In this subsection, we shall extend the analysis to i.i.d. random arrival process where the bit arrivals $\mathcal B_m\in [0, B_{max}]$ is generated by a general distribution $\Pr(\mathcal B)$. The main results are summarized below.

\begin{Theorem}[Performance Bound of Average Delay and Power Consumption for Random Arrivals]\label{the_rbound}{\red For sufficiently small $\Delta t$,} the discrete-time average end-to-end delay $\bar d_{iid}(\Omega^*)$ and average power consumption $\bar g_{iid}(\Omega^*)$ for DBP $\Omega^*$ under i.i.d. arrival process $\{\mathcal B_m\}$ are given by:
\begin{equation}\label{eq_dupr}
\overline d_{iid} \left( {{\Omega ^*}} \right) \le \,\frac{1}{{\overline B_{iid} }}\,\frac{1}{T}\mathbb{E}[{\overline {{J^U}} ^*}(\mathcal{B} + {L^*_{max}},0)]\,\, + \,\mathcal O\left( {{\Delta t}} \right)
\end{equation}
\begin{equation}\label{eq_plowr}
\overline g_{iid} \left( {{\Omega ^*}} \right) \ge \frac{1}{T} \mathbb{E}[{\overline {{J^g}} ^*}(\mathcal B \,,0)]+ \mathcal O\left( {{\Delta t}} \right)
\end{equation}
where $L_{max}^*$ is given by the fixed point of equation (\ref{eq_L}) (with $B=B_{max}$), ${\overline B_{iid}}$ is the average bit arrival rate (bits per seconds) given by ${{\bar B}_{iid}}= {\mathbb E}[{{\cal B}}]/T$ and the expectation is taken w.r.t. the i.i.d. arrival process $\{\mathcal B_m\}$. ~\hfill\IEEEQED
\end{Theorem}
\begin{proof}
Please refer to Appendix G for the proof.
\end{proof}
\begin{Corollary}[Asymptotic Power-Delay Tradeoff of DBP for Random Arrivals at Small Delay Regime]\label{cor_orderr}
{\red For sufficiently small $V$,} the asymptotic power-delay tradeoff of DBP {\blue in the VCTS} under random arrivals $\{\mathcal B_m\}$ is given by:
\begin{equation}\label{eq_dorder}
\overline d_{iid} \left( {{\Omega ^*}} \right) = \mathcal{O} \left( {\mathbb E\left[ {\frac{{{\mathcal B^2}}}{{\log ({\mathcal B }\mathbb{E}[F _{{\psi ^2};{s^2}}^{ - 1}({\varepsilon})]/V)}}} \right]} \right)
\end{equation}
\begin{equation}\label{eq_gorder}
\overline g_{iid} \left( {{\Omega ^*}} \right)  =\Omega \left( {\mathbb E\left[ {\left( {\frac{{{\mathcal B }}}{V} + {P_{cct}}} \right)\frac{\mathcal B}{{ \log ({\mathcal B }\mathbb{E}[F _{{\psi ^2};{s^2}}^{ - 1}({\varepsilon})]/V)}}} \right]} \right)
\end{equation}
\end{Corollary}
\begin{proof}The above results can be obtained in a similar way as Corollary \ref{cor_order} based on Theorem \ref{the_rbound}.
\end{proof}

\section{Simulation Results and Discussions}

In this section, we shall compare our proposed DBP with two reference baselines, namely the CSIT-only control policy (baseline 1) and No-CSIT policy (baseline 2). The baseline 1 policy allocates the rate and power to optimize the PHY throughput based on CSIT only. The baseline 2 policy always transmit with uniform power and fixed rate. In the simulation, we consider both deterministic and random arrivals. The OFDM systems has 1024 subcarriers with total bandwidth 10MHz. The scheduling slot duration $\Delta t$ is 5msec. {\red We simulate $10^6$ scheduling slot to evaluate the average power and delay for different parameter $V$. The dashed lines that pass through the simulation point for DBP algorithm represents the analytical results in \eqref{eq_dorde} and \eqref{eq_gorde}.}

\begin{figure}[t!]
\centering
\includegraphics[width=.5\columnwidth]{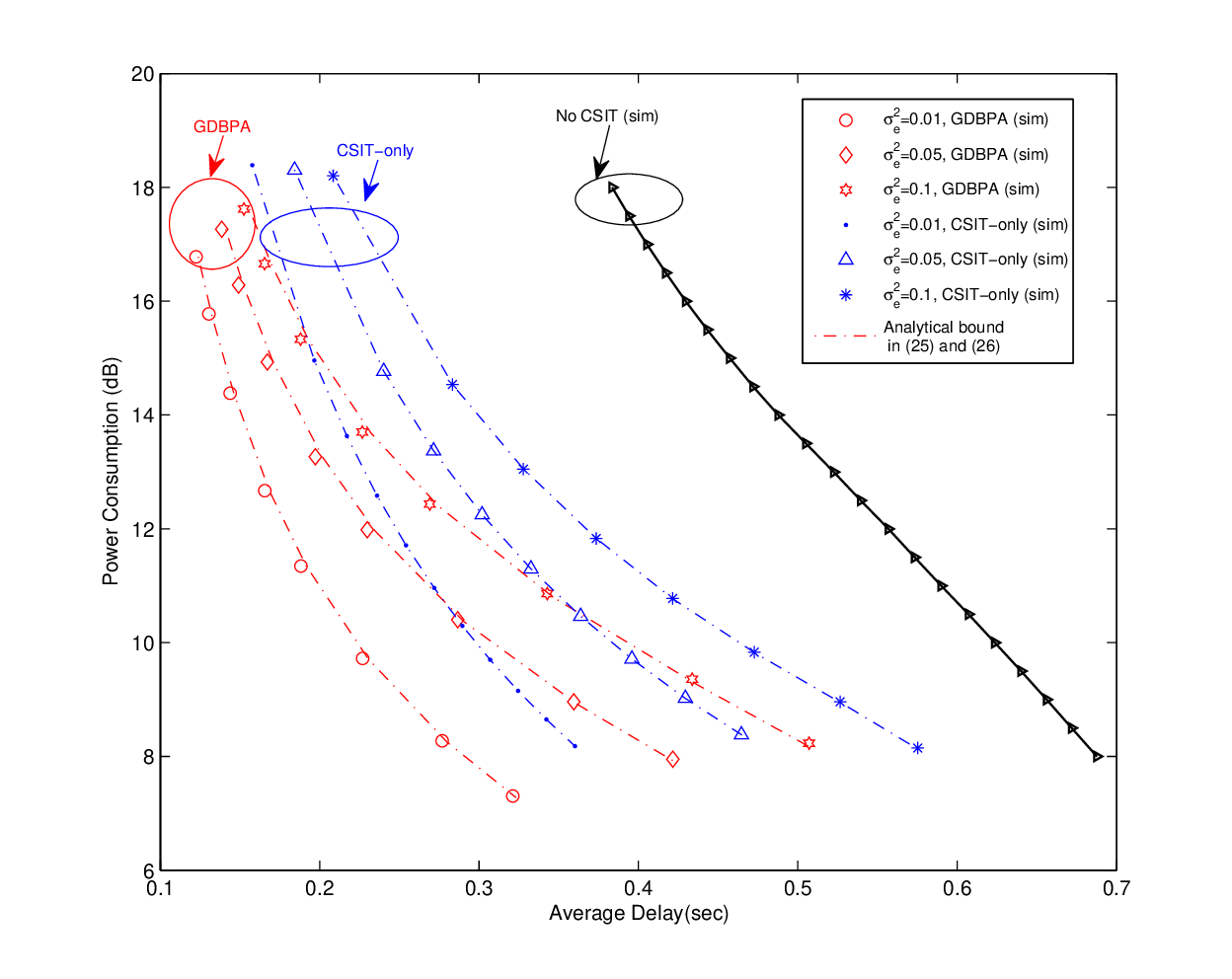}
\caption{{\red Power-delay tradeoff for OFDM link at different CSIT errors for deterministic arrival. CSIT error variance ($\sigma_e^2 = 0.01, 0.05, 0.1$), traffic loading $\bar B$ = 1k nats/slot, $T$=100msec, $\Delta t$ = 5msec, $P_{cct}$ = 0\%, $n_F$=1024, $N_d$ =16 and target PER =0.01. Note that $V$ is a parameter that determines the tradeoff between power and delay. For example, $V=1,2,4, 6,10, 15,25,40$ at the marks along the DBP curves. } \black } \label{fig_pd_sig} \black
\end{figure}

\black

Fig. \ref{fig_pd_sig} illustrates the power-delay tradeoff at different CSIT errors $\sigma^2_e=0.01,0.05,0.1$. It can be observed that DBP simulation results match the performance bounds in Theorem \ref{the_dpbound} quite closely. In addition, it is obvious that DBP has significant gain compared with the CSIT-only policy and No-CSIT policy. {\red It can be observed that the simulation points match with the analytical results very well for small $V$ (which corresponds to small delay regime).} As the CSIT error $\sigma^2_e$ gets smaller, the power-delay curve has steeper slope, which means better tradeoff. The performance gap between DBP and other policies increases at small delay and small CSIT error regime.

%

\begin{figure}
\begin{center}
  \subfigure[ $\sigma_e^2=0.1$ ]
  {\resizebox{8.4cm}{!}{\includegraphics{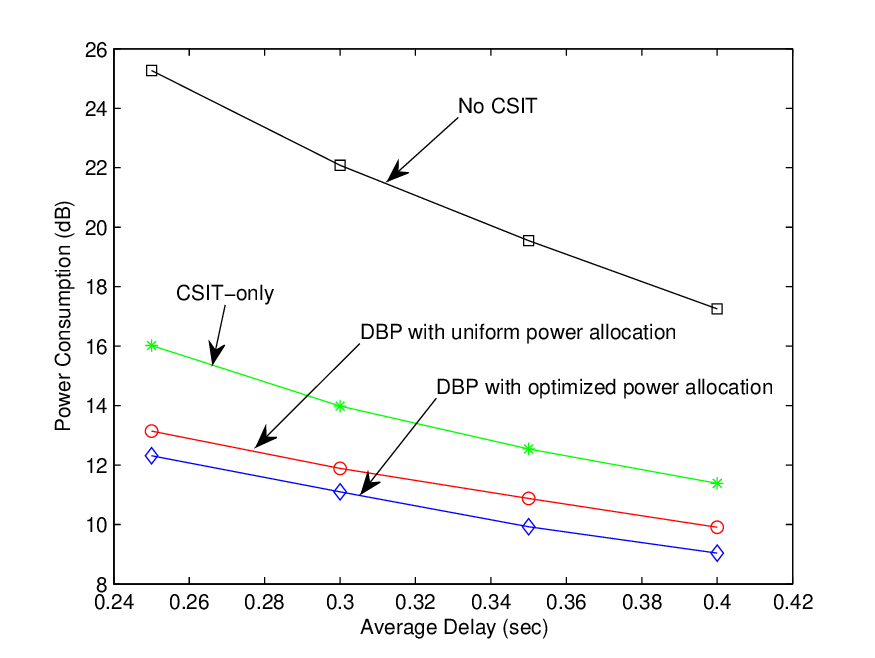}}}
  \subfigure[ $\sigma_e^2=0.15$]
  {\resizebox{8.4cm}{!}{\includegraphics{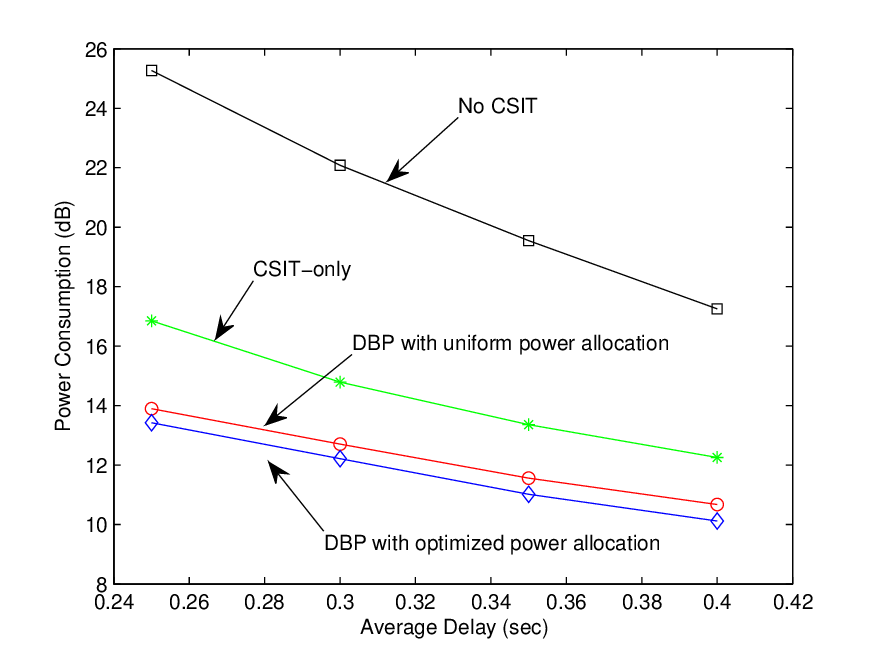}}}
  \end{center}
    \caption{ \red Average Delay versus power consumption at different CSIT errors for deterministic arrival: traffic loading $\bar B$ =  1k nats/slot, $T$=100msec, $\Delta t$ = 5msec, $P_{cct}$ = 0, $n_F$=1024, $N_d$ =16, and target PER =0.01. }
    \label{fig_gradient} \black
\end{figure}

\black


\begin{figure}[h!]
\centering
\includegraphics[width=.5\columnwidth]{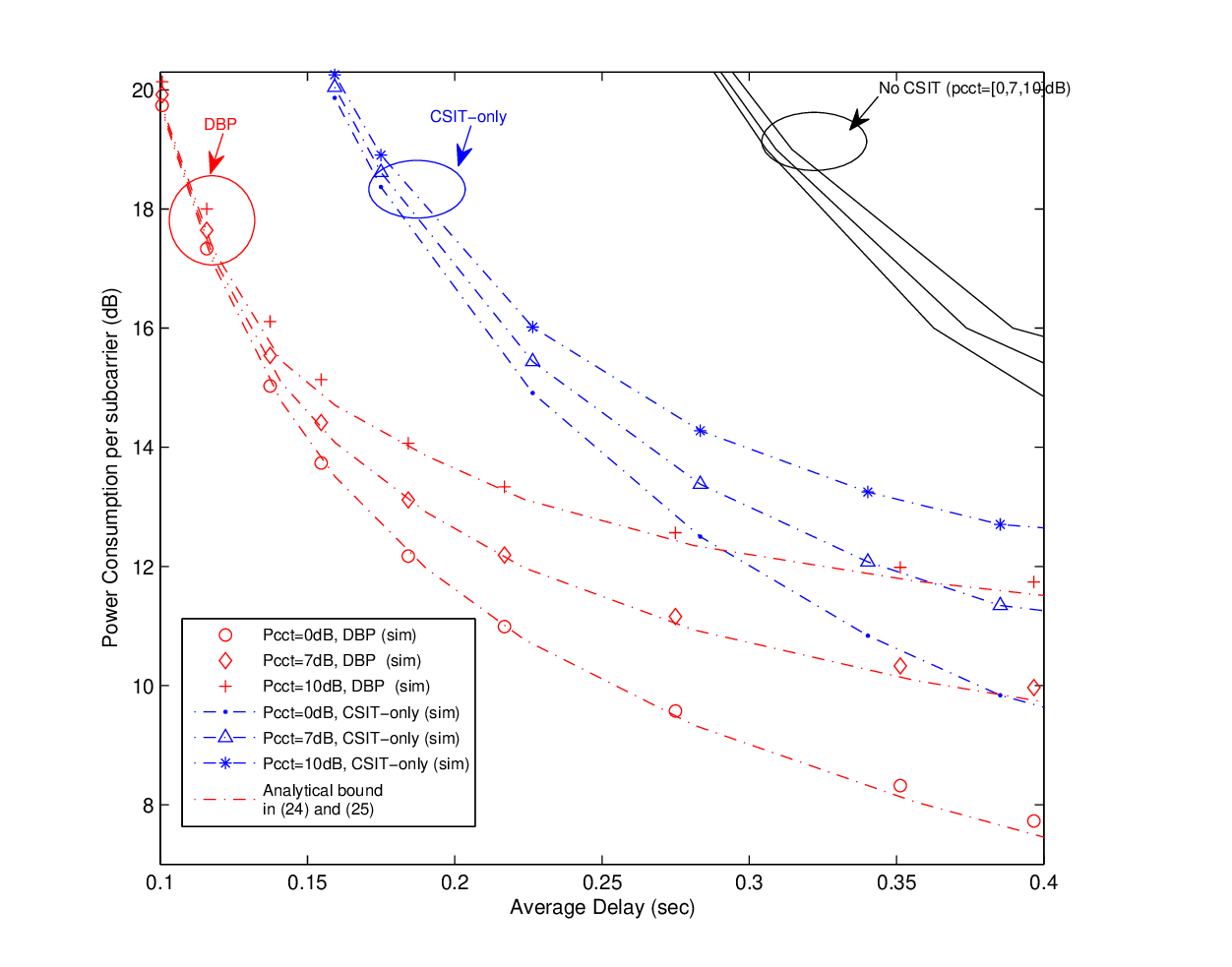}
\caption{{\red Power-delay tradeoff for OFDM link at different $P_{cct}$ for deterministic arrival. Traffic loading $\overline  B$ =  0.8k nats/sec, $T$=100msec, $\Delta t$ = 5msec, $P_{cct}$ = [0,7,10]dB, $\sigma^2_e=0.05$, $n_F$=1024, $N_d$ =16 and target PER =0.01. Note that $V$ is a parameter that determines the tradeoff between power and delay. For example, $V=0.5,1,2,4, 6,10, 15,25,40$ at the marks along the DBP curves. } \black} \label{fig_pd_pcct} \black
\end{figure}

\black

{\red Fig. \ref{fig_pd_pcct} shows that power-delay tradeoff with different circuit power consumption $P_{cct}$. It can be observed that the effect of $P_{cct}$ is significant when $P_{cct}$ is non-negligible from the total power consumption, especially in the large delay regime. Similarly, the DBP has significant gain compared with the CSIT-only policy and No-CSIT policy.}


%

\begin{figure}
\begin{center}
  \subfigure[ Deterministic arrivals ]
  {\resizebox{8.4cm}{!}{\includegraphics{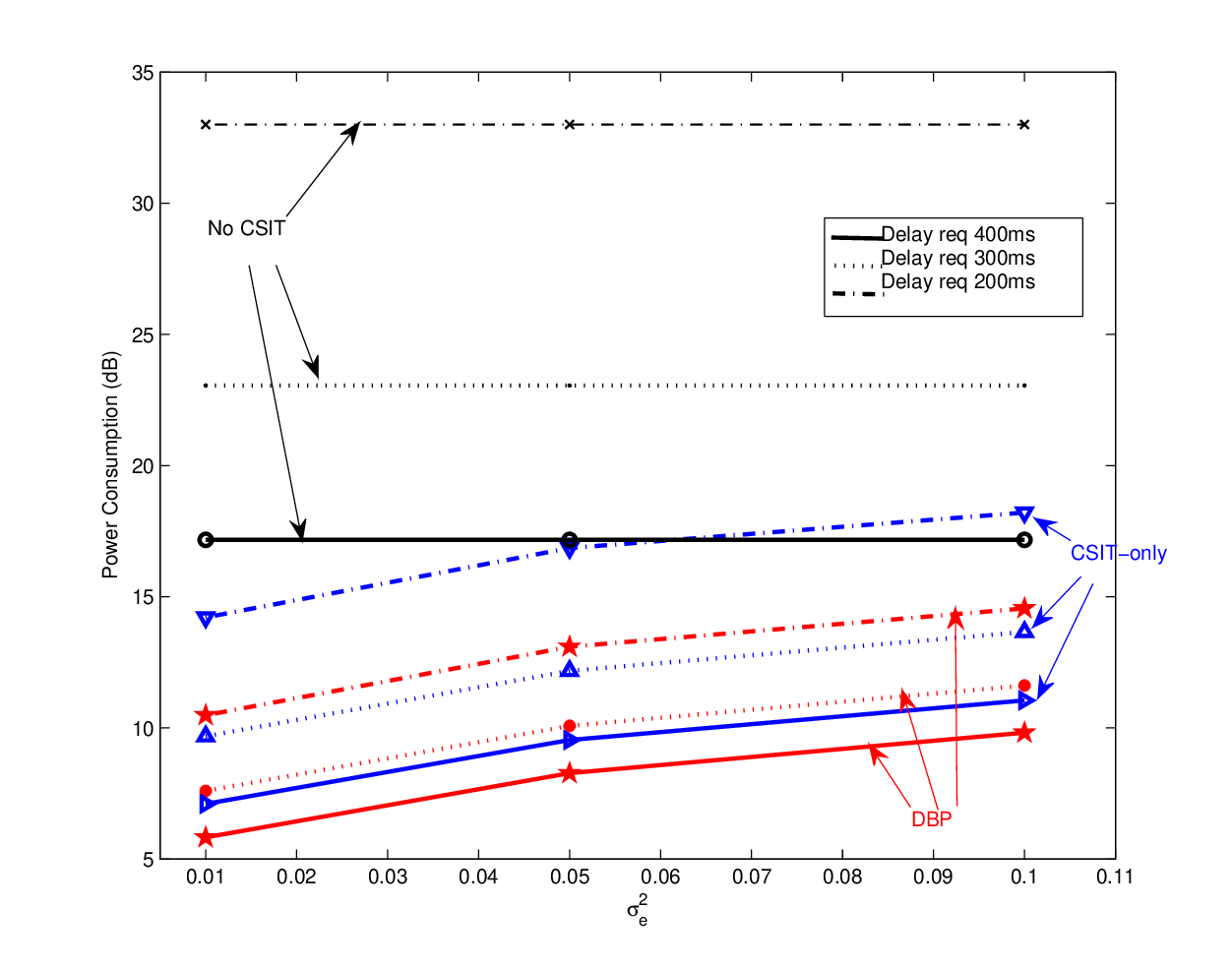}}}
  \subfigure[ Random arrivals ]
  {\resizebox{8.4cm}{!}{\includegraphics{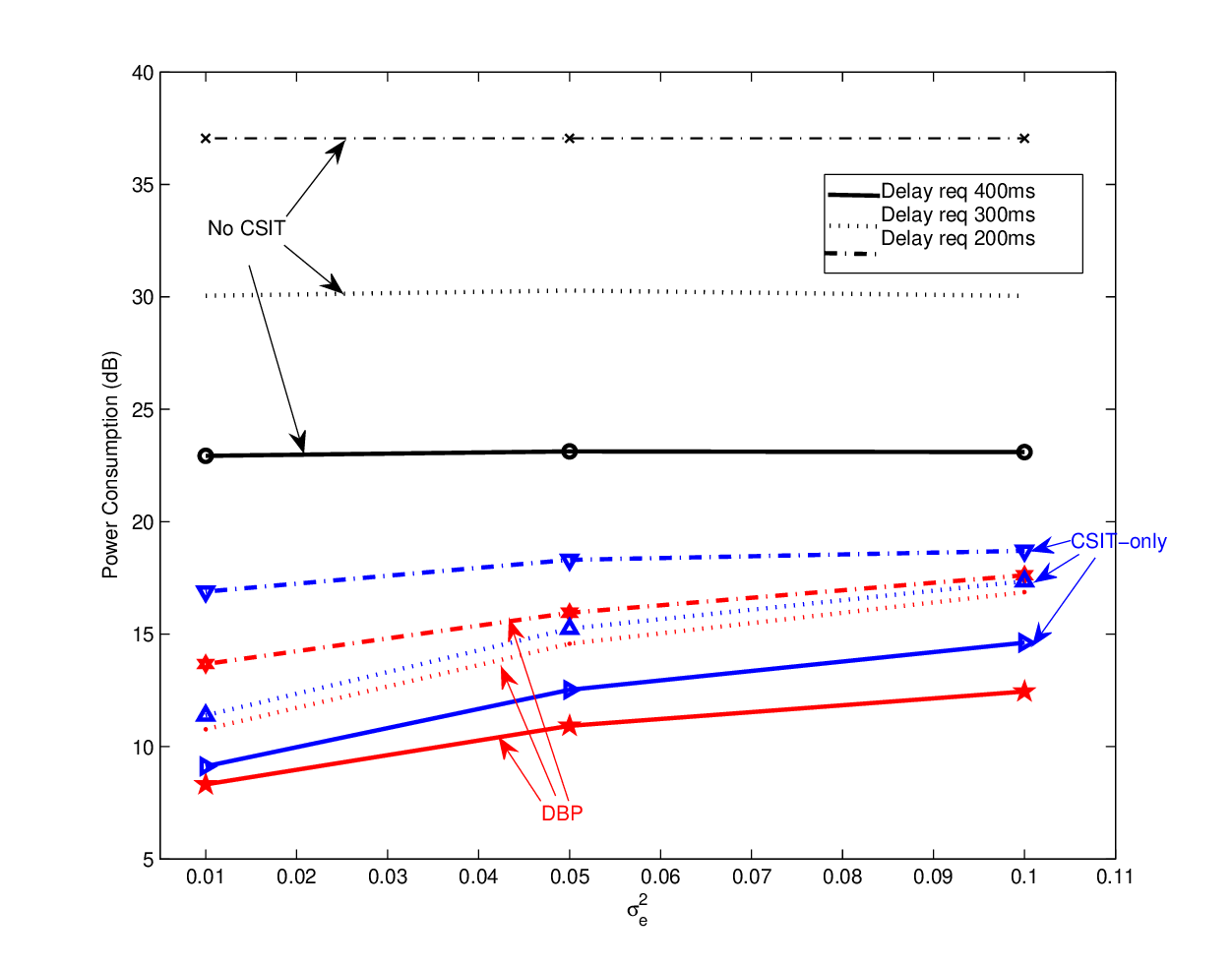}}}
  \end{center}
    \caption{{\red Power consumption versus CSIT errors at different delay requirements: (a) Deterministic arrivals; (b) Random arrivals, traffic loading $\bar B$ =  10k nats/slot, $T$=100msec, $\Delta t$ = 5msec, $P_{cct}$ = 10\%, $n_F$=1024, $N_d$ =16, and target PER =0.01.} \black }
    \label{fig_psig} \black
\end{figure}

\black


Fig. \ref{fig_psig} illustrates the power consumption versus CSIT errors at different delay requirements. It can be seen that when the CSIT error $\sigma^2_e$ increases, the minimum required power for satisfying the delay requirement increases. In addition, the performance gain of the DBP decreases as the CSIT error increases.

\section{Conclusion}
In this paper, we consider a tradeoff of power-delay in point-to-point OFDM systems with imperfect CSIT and non-ideal circuit power. Using Lyapunov optimization framework, we derive a dynamic backpressure algorithm (DBP), which adapts the rate and power based on the instantaneous CSIT and QSI. To study how the CSIT quality and circuit power affects the power-delay tradeoff, we introduce a {\em virtual continuous time system} and derived an asymptotically accurate the power-delay bounds at small delay regime. We show that despite imperfect CSIT quality and non-ideal circuit power, the average power of the DBP policy scales with delay $(D)$ as $\mathcal O(D \exp(1/D))$. The impact of CSIT quality and circuit power appears in the coefficients of the scaling law.


\appendices

\section*{Appendix A: Proof of Lemma \ref{lem_hjb}}\label{app_hjb}
{\red In order to obtain the continuous-time area under the queue trajectory in the VCTS, we shall use the principle of {\em divide-and-conquer} from the definition of $\overline {{J^U}} (\widetilde U,t)$ in \eqref{def_JU} and $\overline {{J^g}} (\widetilde U,t)$ in \eqref{def_Jg}. Specifically, we have}
\[\overline {{J^U}} \left( {\widetilde U,\,t} \right){\mkern 1mu}  = {\mkern 1mu} {\mkern 1mu} \mathbb E\left[ {\int_t^{t + {t_\Delta }} {\widetilde U(s){\mkern 1mu} ds + \overline {{J^U}} ({{\widetilde U}},t + t_\Delta){\mkern 1mu} {\mkern 1mu} \left| {\,\widetilde U(t)} \right.} } \right]{\mkern 1mu} {\mkern 1mu} \]

Using Taylor expansion $\overline {{J^U}} ({\widetilde U},t + {t_\Delta }) = \overline {{J^U}} (\widetilde U,t) - \frac{{\partial \overline {{J^U}} }}{{\partial \widetilde U}}(1 - \varepsilon ){r^*}( \cdot ){t_\Delta }\, + \,\,\frac{{\partial \overline {{J^U}} }}{{\partial t}}{t_\Delta }\,\, + \mathcal O({t_\Delta }^2)$. For small value of $t_\Delta$, $\widetilde U(s)$ is assumed to be fixed to $\widetilde U(t)$ for $s\in[t,t+t_\Delta]$. Thus, removing $\overline {{J^U}} (\widetilde U,t)$ on both sides, we have
\[0 = \mathbb E\left[ {\widetilde U(t){t_\Delta } - \frac{{\partial \overline {{J^U}} }}{{\partial \widetilde U}}(1 - \varepsilon ){r^*}( \cdot ){t_\Delta }\, + \,\,\frac{{\partial \overline {{J^U}} }}{{\partial t}}{t_\Delta } + \mathcal O({t_\Delta }^2)\,\,\,\left| {\widetilde U(t)} \right.} \right].\,\,\]

Dividing by $t_\Delta$ and taking the limit $t_\Delta \rightarrow 0$ gives (\ref{eq_hjb1}). (\ref{eq_hjb2}) is also obtained by a similar procedure from the definition of $\overline {{J^g}} (\widetilde U,t)$ in \eqref{def_Jg}.

\section*{Appendix B: Proof of Theorem \ref{the_1}}\label{app_the_1}
{\red Similar to Appendix A, using the principle of divide and conquer from the definition of $\overline{j^U}(U_0)$ in (\ref{eq_jUdef}) and $\overline{j^g}(U_0)$ in (\ref{eq_jgdef}), we shall have the following recursive equations in discrete-time systems:
\begin{equation}\label{jU_rec}
\overline {j_{}^U} \left( {{U_k}} \right) = \mathbb E\left[ {{U_k}\Delta t\,\, + \,\,\,\overline {j_{}^U} \left( {{U_k}\,\, - \,\,{r^*}(k)(1 - {\varepsilon})\Delta t} \right)\left| {{U_k}} \right.} \right]\,,\,\,k = 0,\, \cdots \,,\frac{T}{\Delta t} - 1
\end{equation}
\begin{equation}\label{jg_rec}
\overline {j_{}^g} \left( {{U_k}} \right) = \mathbb  E\left[ {g\left( {r(k)\,\,,\,\chi (k)} \right)\Delta t\,\, + \,\,\,\overline {j_{}^g} \left( {{U_k}\,\, - \,\,{r^*}(k)(1 - {\varepsilon})\Delta t} \right)\left| {{U_k}} \right.} \right]\,,\,\,k = 0,\, \cdots \,,\frac{T}{\Delta t}- 1
\end{equation}
where $r^*(k)$ is given by the DBP in (\ref{eq_ropt2}). Since $\{\overline {j^U},\overline {j^g}\}$ and $\{\overline {J^U},\overline {J^g}\}$ satisfy the discrete time and continuous time recursive equations, respectively, we only need to show they are different in $\mathcal{O}(\Delta t)$.

We then discuss the property of $\overline {J^U}^*(\widetilde U,t)$.} Note that $\overline {J^U}^*(\widetilde U,t)$ satisfies the PDE in (\ref{eq_hjb1}), multiplying $\Delta t$ in both sides of (\ref{eq_hjb1}), we have:
\begin{equation}\label{eq_PDEd}
\mathbb E\left[ {\widetilde U(t)\Delta t + \,\,\frac{{\partial \overline {{J^U}} }}{{\partial \widetilde U}}\left\{ { - {r^*}( \cdot )(1 - \varepsilon )} \right\}\Delta t + \,\frac{{\partial \overline {{J^U}} }}{{\partial t}}\Delta t\,\,\left| {\widetilde U(t)} \right.} \right] = 0.
\end{equation}

For simplicity, let $\overline {{r^*}}  = \mathbb E[{r^*}(\widetilde U(t),\,\,\widetilde {\bf H}(t))|\widetilde U(t)]\,$. Since $\overline {{r^*}}=0$ if $\widetilde U(t)=0$, by Taylor expansion on $\overline {{J^U}} (\widetilde U - \overline {{r^*}} (1 - \varepsilon )\Delta t,\,\,t + \Delta t)$, we have:
\begin{equation}\label{eq_tal}
 \begin{array}{l}
 \overline {{J^U}} (\widetilde U - \overline {{r^*}} (1 - \varepsilon )\Delta t,t + \Delta t) - \overline {{J^U}} (\widetilde U,t)= \frac{{\partial \overline {{J^U}} }}{{\partial \widetilde U}}( - \overline {{r^*}} (1 - \varepsilon ))\Delta t + \frac{{\partial \overline {{J^U}} }}{{\partial t}}\Delta t + \mathcal O(\Delta {t^2}).
 \end{array}
\end{equation}

By substituting (\ref{eq_tal}) into (\ref{eq_PDEd}):
\begin{equation}\label{eq1}
\overline {{J^U}} (\widetilde U\,,t) = \,\,\mathbb E\left[ {\widetilde U(t)\Delta t + \,\,\overline {{J^U}} (\widetilde U - \overline {{r^*}} (1 - \varepsilon )\Delta t,\,\,t + \Delta t)\,\left| {\widetilde U(t)} \right.} \right] - \,\,\mathcal O(\Delta {t^2}).
\end{equation}

{\red Let $t=k \Delta t$, $\widetilde U(k\Delta t)=\widetilde U_k$ and $\overline {J_{k}^U} ({\widetilde U_k}) = \,\,\overline {{J^U}} (\widetilde U\,,k\Delta t)$, {\red i.e., sampling at time $t=k \Delta t$},
\[\overline {J_k^U} ({\widetilde U_k}) = \,\mathbb E\left[ {{{\widetilde U}_k}\Delta t + \,\,\overline {J_{k + 1}^U} ({{\widetilde U}_{k + 1}})\,\left| {{{\widetilde U}_k}} \right.\,} \right] - \,\,\mathcal O(\Delta {t}^2).\]

{\red Compare to the discrete time recursive equation in (\ref{jU_rec}),} $\overline {J_k^U} (\widetilde U_0,0)$ satisfies it up to $\mathcal O(N\Delta t^2)=\mathcal{O}(\Delta t)$. As a result, the solution of continuous time PDE $\overline {J^U}^*(U_0,0)$ will be different from the actual $\overline{j^{U}}(U_0)$ by at most $\mathcal O(\Delta {t})$, i.e. $\overline{j^{U}}(U_0)= \overline {J^U}^*(U_0,0) + \mathcal O(\Delta t)$. The case of the average energy consumption is obtained similarly. }

\section*{Appendix C: Proof of Theorem \ref{the_bounds}}\label{app_bounds}
{\red In the proof, we shall fist derive the bounds for the average departure rate and power consumption of the DBP policy, which is asymptotically tight at small delay regime (small $V$). Based on these bounds, we can derive the upper and lower bound of the queue trajectory and the corresponding bound of the average unfinished work and energy consumption.

First of all, we have the following lemma on the conditional average rate and power of the DBP policy.}
\begin{Lemma}[Bounds on Conditional Average Policy]\label{lem_approx}
The conditional average rate and transmission power of DBP can be bounded by:
\begin{equation}\label{eq_raprox}
{\overline r _{DBP}}(\widetilde U(t)) \buildrel \Delta \over = \mathbb E\left[ {{r_{DBP}}(\chi )|\widetilde U(t)} \right] \in \,\left[ {{{\overline r }_{low}}(\widetilde U(t)),\,\,{{\overline r }_{up}}(\widetilde U(t))} \right]
\end{equation}
\begin{equation}\label{eq_gaprox}
{\overline g _{DBP}}(\widetilde U(t)) \buildrel \Delta \over = \mathbb E\left[ {{g_{DBP}}(\chi )|\widetilde U(t)} \right] \ge \,\,{\overline g _{low}}(\widetilde U(t))
\end{equation}
where \[{\overline r _{low}}(\widetilde U(t)) = \mathbb E{\left[ {{n_F}\log \left\{ {\frac{{\widetilde U{(t)}(1 - {\varepsilon})f(\epsilon,\sigma_e^2,\mathbf{\hat{H}})}}{V}} \right\}\left| {\widetilde U(t)} \right.} \right]^ + } = {\left[ { {n_F}\log (\widetilde U(t)) + {n_F}\beta } \right]^ + },\]
\begin{equation}\label{eq_rup}
{\overline r _{up}}(\widetilde U(t)) = {n_F}\mathbb E\left[ {{{\left( {\log \widetilde U{(t)}} \right)}^ + } + {{\left( {\log \frac{{(1 - {\varepsilon})f(\epsilon,\sigma_e^2,\mathbf{\hat{H}})}}{V}} \right)}^ + }\left| {\widetilde U(t)} \right.} \right] =  {n_F}{\left[ {\log (\widetilde U(t))} \right]^ + } + {n_F}\beta ',
\end{equation}
and $\overline g_{low}(\widetilde U(t)) = {\left[ {\frac{{{{\widetilde U}(t) }(1 - {\varepsilon}){n_F}}}{V} + \overline g_{DBP}(\widetilde U(t_0)) - \frac{{{{\widetilde U}}(t_0) (1 - {\varepsilon}){n_F}}}{V}} \right]^ + }$. ~\hfill\IEEEQED
\end{Lemma}
\begin{figure}[t!]
\centering
\includegraphics[width=.5\columnwidth]{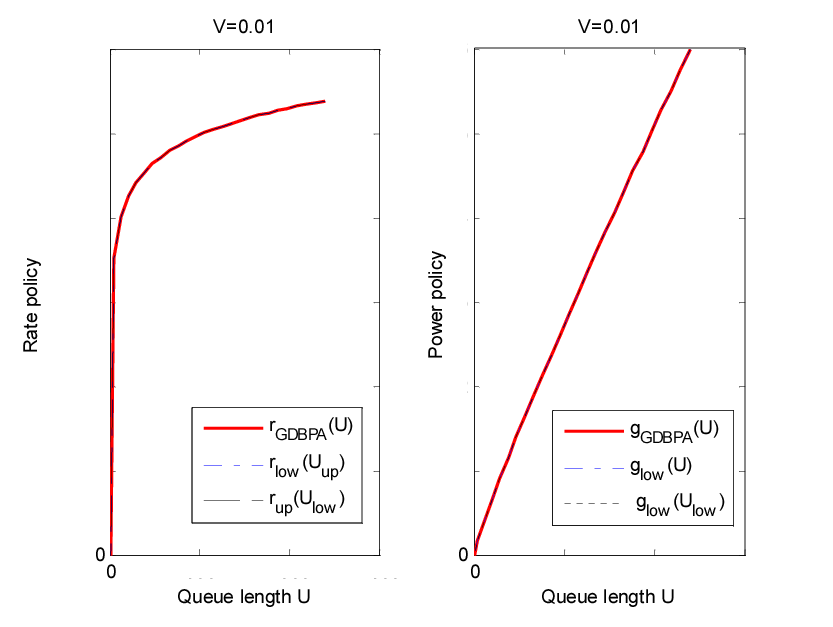}
\caption{Comparison of the actual and approximated rate and power control policies.} \label{fig_Vacc}
\end{figure}

As a result, we can approximate $\overline{r}_{DBP}$ and $\overline{g}_{DBP}$ using $\overline {r}_{low}(U)$ and $\overline{g}_{low}(U)$ with asymptotically small approximation errors at small $V$. Fig. \ref{fig_Vacc} illustrates the accuracy of the approximation.
\begin{proof} {\red The key proof is to find the bound for the average departure rate. Specifically,
from (\ref{eq_ropt2}), using Jensen's inequality yields following inequality, we have the lower bound:}
\begin{equation}\label{eq_rvct}
{\overline r _{DBP}}\left( {\widetilde U(t)} \right) = \mathbb E\left[ {{{\left( {{n_F}\log \left\{ {\frac{{\widetilde U (t)(1 - {\varepsilon})f(\epsilon,\sigma_e^2,\mathbf{\hat{H}})}}{V}} \right\}} \right)}^ + }\left| {\widetilde U(t)} \right.} \right] \ge \,\,{\overline r _{low}}\left( {\widetilde U(t)} \right).
\end{equation}
Similarly, using the fact that $(x + y)^+ \leq (x^+ + y^+)$, {\red we have the upper bound}:
\[{\overline r _{DBP}}\left( {\widetilde U(t)} \right) = {n_F}\mathbb E\left[ {{{\left( {\log \left\{ {\widetilde U(t)} \right\} + \log \left\{ {\frac{{(1 - {\varepsilon})f(\epsilon,\sigma_e^2,\mathbf{\hat{H}})}}{V}} \right\}} \right)}^ + }\left| {\widetilde U(t)} \right.} \right] \le {\overline r _{up}}(\widetilde U(t)).\]

{\red Now, we shall show the lower bound of ${\overline g _{DBP}}(\widetilde U(t))$.} Let ${f(\epsilon,\sigma_e^2,\mathbf{\hat{H}})}=X$ and $q(x)$ to be the pdf of the random variable $X$, where the randomness is induced by the CSIT $\widehat{\bf{H}}$. From (\ref{eq_Popt}), we have:
\[{\overline g _{DBP}}(\widetilde U(t)) = \mathbb E\left[ {{{\left( {\frac{{\widetilde U(t){n_F}(1 - {\varepsilon})}}{V} - \frac{{{n_F}}}{X}} \right)}^ + } + {P_{cct}}\left| {\widetilde U(t)} \right.} \right] = \int_{\frac{V}{{{{\widetilde U} }(1 - {\varepsilon})}}}^\infty  {\left( {\frac{{{{\widetilde U} }{n_F}(1 - {\varepsilon})}}{V} - \frac{{{n_F}}}{X} + {P_{cct}}} \right)} q(x)dx.\]

$\overline g _{DBP}(\widetilde U(t))$ is a monotonic increasing function of $\widetilde U(t)$. Differentiating both sides w.r.t. $\widetilde U$:
\[\frac{{d{{\overline g }_{DBP}}(\widetilde U)}}{{d\widetilde U}}\, = \,\,\frac{(1 - {\varepsilon}){n_F}}{V}\int_{\frac{V}{{{\widetilde U}(1 - {\varepsilon})}}}^\infty  {q(x)} \,dx\, \le \,\,\frac{(1 - {\varepsilon}){n_F}}{V}\,,\,\,\,\,\,\forall {\overline g _{DBP}}(\widetilde U) \ge 0.\]
As a result, we can construct a lower bound of $\overline g_{DBP}$ by the following ODE:
\begin{equation}\label{eq8}\frac{{d{{\overline g }_{low}}(\widetilde U)}}{{d\widetilde U}} = \left\{ {\begin{array}{*{20}{c}}
   {\frac{ (1 - {\varepsilon}){n_F}}{V},\,\,\,{\rm{if}}\,\,\,\widetilde U\, \in \,\,{\Theta _1}\,}  \\
   {\,\,\,\,\,\,\,\,0\,\,\,\,\,\,\,\,\,\,\,\,\,\,\,\,\,\,,\,\,\,\,\,{\rm{otherwise}}}  \\
\end{array}\,} \right.
\end{equation}
\begin{equation}\label{eq9}{\overline g_{low}}({\widetilde U_0}) = {\overline g_{DBP}}({\widetilde U_0})\end{equation}
where $\Theta_1$ is $\,\left\{ {\widetilde U\left| {{{\overline g _{low}}}(\widetilde U) \ge 0} \right.} \right\}$. ${\overline g _{low}}(\widetilde U)$ always has steeper slope compared with $\overline g_{DBP}(\widetilde U)$ and together with the monotonic increasing property of $\overline g _{DBP}(\widetilde U)$, we could establish the lower bound by solving (\ref{eq8}) and (\ref{eq9}).
\end{proof}

{\red Secondly, based on Lemma \ref{lem_approx}, we shall derive an upper bound on $\overline {J^U}^*(\widetilde U,t)$ by solving the PDE in Lemma \ref{lem_hjb}.} Specifically, the queue dynamics $\widetilde U(t)$ satisfies the ODE in the VCTS in (\ref{eq_ode}):
\begin{equation}\label{eq_odeapx}
\frac{d\widetilde U}{dt} = -\overline{r}_{DBP}(\widetilde U)(1-\varepsilon) \leq -\overline{r}_{low}(\widetilde U)(1-\varepsilon).
\end{equation}
Hence, the upper bound queue trajectory is given by:
\begin{equation}\label{eq_uup}
\widetilde U(t)\,\, \le {\widetilde U_{up}}(t) = {\mkern 1mu} \exp \left[ { - \beta {\kern 1pt} {\kern 1pt}  + {\kern 1pt} {\rm{E}}{{\rm{i}}^{( - 1)}}\left[ {{\rm{Ei}}\left( {\log ({{\widetilde U}_0}) + \beta } \right) - (1 - {\varepsilon}){n_F}{e^{\beta}}t} \right]} \right]\, = \,{\mkern 1mu} y(t;\,\,\beta ),
\end{equation}
Based on (\ref{def_JU}), we have
\[{\overline {{J^U}} ^*}(\widetilde U,\,\,0) \approx \mathbb E\left[ {\int_0^T {\widetilde U(t)dt} \left| {{{\widetilde U}_0}} \right.} \right]\,\, \le \,\int_0^T {{{\widetilde U}_{up}}} (t)dt\,\]
which yields the upper bound of ${\overline {{J^U}} ^*}(\widetilde U,\,\,0)$ in (\ref{eq_JUup}).

{\red Finally, we shall derive a lower bound for $\overline {J^g}^*(\widetilde U,0)$.} Using $\overline r_{up}(\widetilde U(t))$, we can construct a lower bound trajectory by solving $\frac{d\widetilde U_{low}}{dt} = -\overline r_{up}(\widetilde U_{low})(1-\varepsilon)$,
and the solution is given by:
\begin{equation}\label{eq_ulow}
\widetilde U(t)\,\, \ge {\widetilde U_{low}}(t) = {\mkern 1mu} \exp \left[ { - {\beta '}{\kern 1pt} {\kern 1pt}  + {\kern 1pt} {\rm{E}}{{\rm{i}}^{( - 1)}}\left[ {{\rm{Ei}}\left( {\log ({{\widetilde U}_0}) + {\beta '}} \right) -  (1 - {\varepsilon}){n_F}{e^{{\beta '}}}t} \right]} \right]\, = \,{\mkern 1mu} y(t;\,\,\beta '),
\end{equation}

At any time $t \in[0,T]$, we have
\begin{equation}\label{eq_glows}
{\overline g _{low}}({\widetilde U_{low}}(t))\,\,\, \le \,\,{\overline g _{low}}(\widetilde U(t))\,\, \le \,\,{\overline g _{DBP}}(\widetilde U(t)).
\end{equation}
Hence,
\[{\overline {{J^g}} ^*}(\widetilde U,\,\,0) = \int_0^T {{{\overline g }_{DBP}}(\widetilde U(t))\,\,} dt\,\, \ge \,\int_0^T {{{\left[ {\frac{{{{\widetilde U}_{low}}(t){n_F}(1 - {\varepsilon})}}{V} + \,{{\overline g }_{DBP}}({{\widetilde U}_0}) - \frac{{\widetilde U_0 {n_F}(1 - {\varepsilon})}}{V}} \right]}^ + }dt} \]
which yields the lower bound of per-period average energy consumption in (\ref{eq_Jglow}).

\section*{Appendix D: Proof of Lemma \ref{lem_left}}\label{app_left}
{\red Using the theory developed for the continuous time model in Section IV-A, it is enough to show that $L_m$ for the VCTS is bounded by $L^*$ for all $m$. First of all, we shall show that $L^*$ is unique.}

{\red Specifically, the leftover bits at the end of the $m$-th period $L_m$ is:
\begin{equation}\label{eq17}
\underline {{l_m}}  \le {L_m} \le \overline {{l_m}}
\end{equation}
where $\overline {{l_m}}$ and $\underline {{l_m}}$ are the leftover bits of $m$-th period of approximated queue trajectory $\widetilde U_{up}(t)$ and  $\widetilde U_{low}(t)$ in (\ref{eq_uup}) and (\ref{eq_ulow}), respectively. Recall that $\widetilde U_{low}(t)\le\widetilde U(t) \le\widetilde U_{up}(t)$. Since the unfinished work at the start epoch of $m$-th period is the summation of arriving bits $B$ and the leftover bits of the previous $(m-1)$-th period, $\overline {{l_m}}$ and $\underline {{l_m}}$ satisfy the followings:
\[\overline {{l_m}}  = \overline f \left( {B + \overline {{l_{m - 1}}} } \right)\,\,\,\,\,\,{\rm{and}}\,\,\,\,\,\,\underline {{l_m}}  = \underline f \left( {B + \underline {{l_{m-1}}} } \right)\]
where $\overline f (x) = \exp \big[ { - \beta + {\rm{E}}{{\rm{i}}^{( - 1)}}[ {{\rm{Ei}}( {\log (x) + \beta} ) - { {n_F}(1 - \varepsilon )}{e^{\beta}}T} ]} \big]$, and $
\underline f (x) = \text{exp} \big[  - {\beta '} + {\rm{E}}{{\rm{i}}^{( - 1)}}[ {\rm{Ei}}( \log (x)+ {\beta '} ) - {{n_F}(1 - \varepsilon )} {e^{{\beta '}}}T ] \big]$.
Both $\overline f (x)$ and $\underline f (x)$ are the increasing functions of $x$ since the exponential integral $\rm{Ei}$ is increasing function.
Note that the slope of $\overline f(x)$ is given by:
\[\frac{{d\overline f(x)}}{{dx}}\, = \,\frac{{\,{\rm{E}}{{\rm{i}}^{( - 1)}}[{\rm{Ei}}\left( {\log (x) + \beta } \right) -  (1 - {\varepsilon}){n_F}{e^{\beta}}T]}}{{\log (x) + \beta  }}\, < 1\,,\,\,\,\,\,{\rm{for}}\,\,{\rm{all}}\,\,x>0.\]
Therefore, there exists a unique crossing point between $y=x$ and $y=\overline f(B+x)$ and this proved the existence and uniqueness of $L^*$. }

{\red Finally, we shall try to prove that $\sup_{m\in\mathbb R^+}L_m$  is bounded by $L^*$.} We first claim that $\overline{l_m} \leq L^*$ for all $m$. From (\ref{eq_L}) and $\overline l_m  = \overline f \left( {B + \overline l_{m - 1} } \right)\,$, we have:
\begin{equation}\label{eq14}
{\rm{Ei}}\left( {\log ({L^*}) + \beta} \right) = {\rm{Ei}}\left( {\log (B + {L^*}) + \beta} \right) -  (1 - {\varepsilon}){n_F}{e^{\beta}}T
\end{equation}
\begin{equation}\label{eq15}
{\rm{Ei}}\left( {\log ({\overline l_m}) + \beta} \right) = {\rm{Ei}}\left( {\log (B + {\overline l_{m - 1}}) + \beta} \right) - (1 - {\varepsilon}){n_F}{e^{\beta}}T
\end{equation}
Subtracting (\ref{eq15}) from (\ref{eq14}), we have
\begin{equation}\label{eq16}
\int_{\log ({\overline l_m}) + \beta  }^{\log ({L^*}) + \beta  } {\frac{{{e^x}}}{x}\,} dx\, = \int_{\log (B + {\overline l_{m - 1}}) + \beta  }^{\log (B + {L^*}) + \beta  } {\frac{{{e^x}}}{x}\,} dx.\,\
\end{equation}
Note that $e^x/x$ is positive for $x>0$, and $\log(\overline l_m)+\beta>0$ due to $\overline f(x)>e^{-\beta}$. If $\{\overline l_m\}$ is not bounded by $L^*$, there exists $m'$ such that $\overline l_{m'-1}\le L^* <\overline l_{m'}$ and $B+\overline l_{m'-1}\le B+L^* <B+\overline l_{m'}$. This $m'$ makes the RHS of (\ref{eq16}) positive while the LHS of (\ref{eq16}) becomes negative which means a contradiction. Thus, $\{\overline l_m\}$ is bounded by $L^*$. As a result, $\sup_{m\in\mathbb R^+}L_m$ is bounded by $L^*$.


\section*{Appendix E: Proof of Theorem \ref{the_dpbound}}\label{app_dpbound}
{\red In the proof, we shall use Little's law \cite{Ross:2003} to derive the average delay and power consumption for the real discrete time system.}

Specifically, under the stationary DBP policy, the system state $\chi(k)$ evolves as an ergodic Markov chain and hence, there exists a steady state distribution $\pi_{\chi}$ such as ${\pi _\chi }\left( {{\chi _0}} \right)\,\, = \,\,{\lim _{k \to \infty }}\Pr \left[ {\chi (k) = {\chi _0}} \right]\,$. Using Little's law and the characteristic of ergodic chain, the average end-to-end delay under the stationary DBP policy $\Omega^*$ is given by:
\[\overline d \left( {{\Omega ^*}} \right) = \frac{1}{{\overline B }}{E_{{\pi _\chi }}}\left[ {U(k)} \right]\, = \frac{1}{{\overline B }}\mathop {\lim }\limits_{K \to \infty } \sum\nolimits_{k = 0}^{K - 1} {\frac{1}{K}\mathbb E\left[ {U(k)\left| {U(0)} \right.} \right]} \,\]
where the ${E_{{\pi _\chi }}}$ is the expectation w.r.t. the steady state distribution of $U(k)$. Using ergodic theory, we have:
\[\overline d \left( {{\Omega ^*}} \right) = \,\frac{1}{{\overline B \,T}}\mathop {\lim }\limits_{M \to \infty } \frac{1}{M}\sum\nolimits_{m = 0}^{M - 1} {{{\overline {{j^U}} }^*}(U_0^m)} \]
where $U_0^m$ is the initial queue length of the $m$-th period. {\red Note that the unit of $\overline B \,T$ is bits, and the unit of ${{\overline {{j^U}} }^*}$ is bits$\times$seconds, and hence the unit of $\overline d $ is seconds.} As shown in Lemma \ref{lem_left}, the leftover $L_m$ at the end of the $m$-th arrival period is bounded by $L^*$. Hence, $U_0^m$ is upper bounded by $B+L^*$. Using the relationship between the unfinished works in discrete time and continues time in Theorem \ref{the_1}, the discrete time average delay given by:
\[\overline d \left( {{\Omega ^*}} \right) \le \,\,\,\frac{1}{{\overline B T}}{\overline {{j^U}} ^*}(B + {L^*}) = \,\frac{1}{{\overline B T}}{\overline {{J^U}} ^*}(B + {L^*},0)\,\, + \,\mathcal O\left( {{\Delta t}} \right).\]

Similarly, the average power consumption $\overline g ({\Omega ^*})$ is given by:
\[\overline g \left( {{\Omega ^*}} \right) = \frac{1}{T}\mathop {\lim }\limits_{M \to \infty } \frac{1}{M}\sum\nolimits_{m = 0}^{M - 1} {{{\overline {{j^g}} }^*}(U_0^m)} \,\,\, \ge \,\,\frac{1}{T}\,\,{\overline {{j^g}} ^*}(B) = \frac{1}{T}{\overline {{J^g}} ^*}(B,\,\,0)\,\, + \,\mathcal O\left( {{\Delta t}} \right).\]
{\red Note that the unit of $T$ is seconds, and the unit of ${{\overline {{j^g}} }^*}$ is Watt$\times$seconds, and hence the unit of $\overline g $ is Watt.}

\section*{Appendix F: Proof of Corollary \ref{cor_order}}\label{app_order}
{\red From the Theorem \ref{the_dpbound}, we only need to obtain ${\overline {{j^U}} ^*}(B + {L^*})$ and ${\overline {{J^g}} ^*}(B,0)$ for the continuous time model by studying the continuous queue trajectory.}

{\red First of all, we show the asymptotic behavior of $L^*$ as $V$ goes to 0.} Let $x$ to be $\beta$ for simplifying  notation. From (\ref{eq_L}), we have: ${L^*}\, = \,{e^{ - x + {\rm{E}}{{\rm{i}}^{( - 1)}}\left[ {{\rm{Ei}}\left( {\log (B + {L^*}) + x} \right) -  n_F(1-\varepsilon){e^x}} \right]}}$. As $V$ goes to 0, $\beta$ and $x$ increase since $\beta=\Theta(\log(1/V))$. For $x>0$, $e^x$ increases faster than $\mathrm{Ei}(x)$ because
\[{\rm{Ei}}(x)\, = \,\,\log x\,\, + \,\gamma \,\, + \,\sum\nolimits_{k = 1}^\infty  {\frac{{{x^k}}}{{k\,k!}}} \, < \,\sum\nolimits_{k = 1}^\infty  {\frac{{{x^k}}}{{k!}}} \, = {e^x}.\]
Hence, as $V$ approaches to 0, we have $L^* = \Theta\left( e^{-\beta}\right) = \Theta\Big(\frac{V}{{\mathbb{E}[ f(\epsilon,\sigma_e^2,\mathbf{\hat{H}})]}}\Big)$

{\red Secondly, we shall obtain the asymptotic area of buffer trajectory. Specifically, we use Fig. \ref{fig_traj} to illustrate the proof.} The upper bound of queue delay can be derived from the area of $\widetilde U_{up}(t)$ in (\ref{eq_uup}). As shown in (\ref{eq_odeapx}), $d\widetilde U_{up}/dt \propto \log \widetilde U_{up}(t)$, which is decreasing in $t$. Hence, $\widetilde U_{up}(t)$ is a convex function in $t$ and we could upper bound ${\overline {{J^U}} ^*}(B+L^* ,0)$ by the summation of the triangle area (A) and the rectangle area (B) as illustrated in Fig. \ref{fig_traj}(a):
\begin{equation}\label{eq18}
{\overline {{J^U}} ^*}(B+L^* ,0) \le \,\frac{1}{2}(B-L_{\Delta})\,{t_d}\,\, + \,T(L^*+L_\Delta).
\end{equation}
where $L_\Delta=\Theta(e^{-\beta})$ and $t_d$ is the time when $\widetilde U_{up}(t) = L^*+L_\Delta$ which is given by:
\[{t_d}\,\, = \,\frac{{{e^{ - \beta }}}}{{ (1 - {\varepsilon}){n_F}}}\left\{ {{\rm{Ei}}\left( {\log (B + {L^*}) + \beta} \right) - {\rm{Ei}}\left( {\log ({L^*} + {L_\Delta }) + \beta} \right)} \right\} = \Theta \left( {\frac{B}{{\log ({B }\mathbb{E}[ f(\epsilon,\sigma_e^2,\mathbf{\hat{H}})]/V)}}} \right)\]
From (\ref{eq18}), the upper bound of average delay $\overline d(\Omega^*)$ is given by (\ref{eq_dorde}).

\begin{figure}[t!]
\centering
\includegraphics[width=.45\columnwidth]{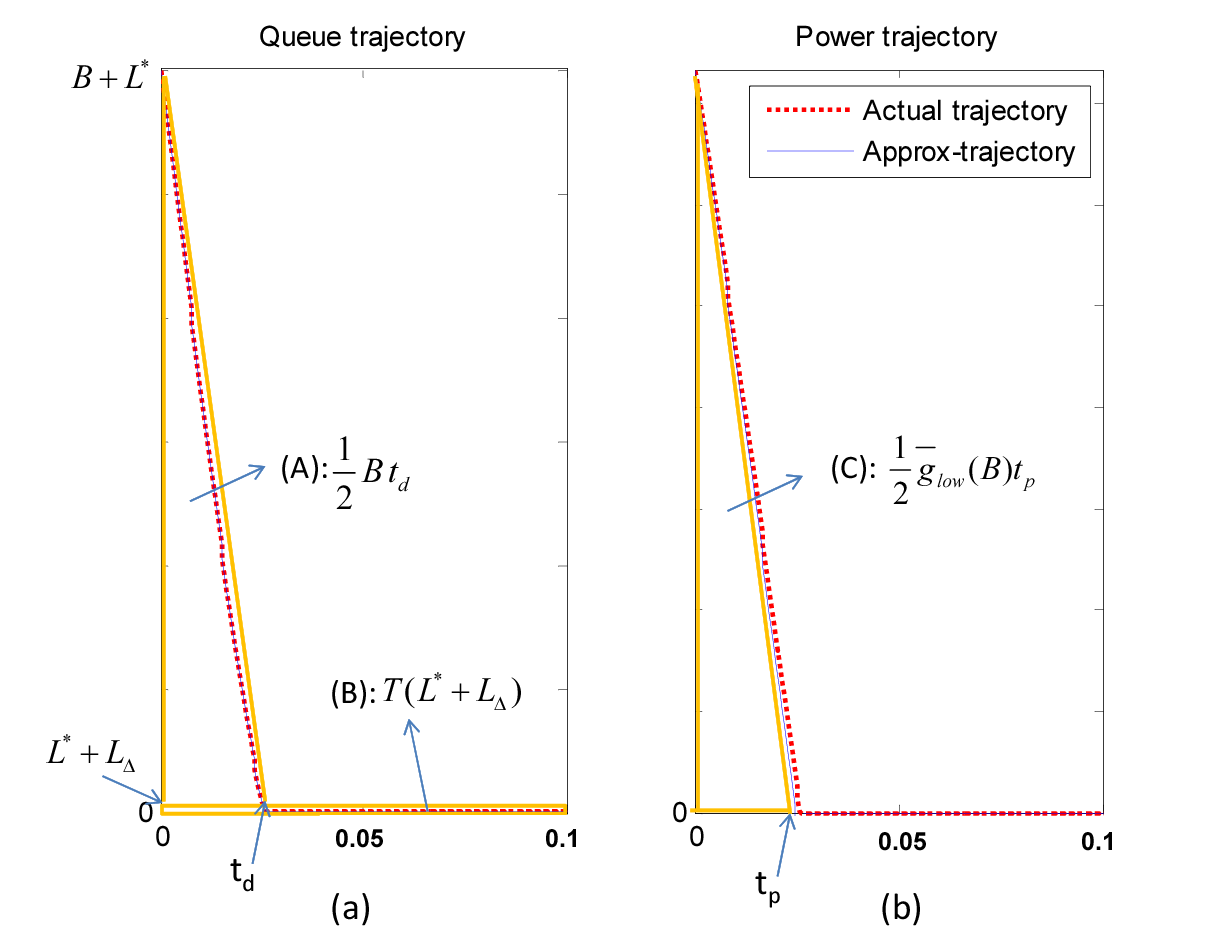}
\caption{Asymptotic upper bound of the area below the queue and power trajectory.} \label{fig_traj}
\end{figure}

{\red Then, we focus on deriving ${\overline {{J^g}} ^*}(B,0)$. Specifically, we use Fig. \ref{fig_traj} to illustrate the proof. We elaborate the asymptotic expression using $\overline g_{low}(\widetilde U_{low}(t))$ where $\widetilde U_{low}(t)$ is the lower bound trajectory derived using $\overline r_{up}$ in (\ref{eq_ulow}). According to (\ref{eq_glows}), $\overline g_{low}$ is a lower bound of the actual power trajectory and it is a convex function of $t$. As a result, the lower bound of the transmission energy (area of $\overline g_{low}$) is given by the triangle (C) of the tangent line of the $\overline g_{low}(B)$ as illustrated in Fig. \ref{fig_traj}(b). Let $t_p$ to be the time when the tangent line touches zero which is given by
\[{t_p} = \frac{{{{\overline g }_{low}}(B)}}{{{{\left. { - \frac{{d{{\overline g }_{low}}(\widetilde U(t))}}{{dt}}} \right|}_{\widetilde U = B}}}} = \frac{{\mathbb E\left[ {{{\left[ {\frac{{{B }{n_F}(1 - {\varepsilon})}}{V} - \frac{{{n_F}}}{{f(\epsilon,\sigma_e^2,\mathbf{\hat{H}})}}} \right]}^ + } + {P_{cct}}} \right]}}{{{{\left. {\frac{1}{V}{{\overline r }^*}(\widetilde U)} \right|}_{\widetilde U = B}}}} = \Theta \left( {\frac{B}{{ \log ({B }\mathbb{E}[f(\epsilon,\sigma_e^2,\mathbf{\hat{H}})]/V)}}} \right).\]
Note that $t_p$ means the transmission time for the burst transmission. On the other hand, we get the same order of $t_p$ in the case of that $\overline g_{low}$ is a concave function \footnote{The $t_p$ of the concave case is obtained from $\overline g_{low}(\widetilde U_{low}(t_p))=0$.}. Hence, the asymptotic lower bound for the total energy consumption in one period is given by:
\[\overline {{J^g}}^*(B,\,\,0)  \ge \,{\mkern 1mu} \,\frac{1}{2}{\overline g _{low}}(B){\mkern 1mu} {t_p}\,\, \approx \,\,\mathbb E\left[ \left( {\frac{{{B }{n_F}(1 - {\varepsilon})}}{V} - \frac{{{n_F}}}{{f(\epsilon,\sigma_e^2,\mathbf{\hat{H}})}^ + }} \right) \right]{t_p} + {P_{cct}}{t_p}{\mkern 1mu} \]
which yields the lower bound of average power consumption in (\ref{eq_gorde}).}


\section*{Appendix G: Proof of Theorem \ref{the_rbound}}\label{app_rbound}
{\red To derive the bound for the random bits arrival, we shall first derive the upper bound of the steady state leftover queue length of a period (where the lower bound is 0 obviously.)}

{\red First of all, from the results of deterministic arrivals in Theorem \ref{the_dpbound},} the average end-to-end delay $\widetilde d(\Omega^*)$ and the average power consumption $\widetilde g(\Omega^*)$ for random arrivals are given by:
\[\widetilde d\left( {{\Omega ^*}} \right) = \,\,\frac{1}{{\overline {\mathcal B} \,T}}\mathbb E\left[ {{{\overline {{J^U}} }^*}(\mathcal B + \mathcal L,0)} \right]\,\, + \,\mathcal O\left( {\Delta t}\right)\]
\[\widetilde g\left( {{\Omega ^*}} \right) = \,\,\frac{1}{{\,T}}\mathbb E\left[ {{{\overline {{J^{g}}} }^*}(\mathcal B + \mathcal L,\,\,0)} \right]\,\, + \,\mathcal O\left( {\Delta t}\right)\]
where $\mathcal L$ is the steady state leftover queue length of a period for random arrivals. {\red From Lemma \ref{lem_left}, it is obvious that $\mathrm L^*_{max}$ is the upper bound of $\mathcal L$.} Then we have the upper bound and lower bound of average per-period unfinished work and energy consumption in VCTS:
$\mathbb E\left[ {{{\overline {{J^U}} }^*}(\mathcal B + \mathcal L,\,\,0)} \right] \le {\mkern 1mu} {\mkern 1mu} {\mkern 1mu} \mathbb E{\mkern 1mu} \left[ {{{\overline {{J^U}} }^*}(\mathcal B + {\rm{L}}_{\max }^*,0)} \right]$ and $\mathbb E\left[ {{{\overline {{J^g}} }^*}(\mathcal B + \mathcal L,\,0)} \right] \ge {\mkern 1mu} {\mkern 1mu} {\mkern 1mu} \mathbb E{\mkern 1mu} \left[ {{{\overline {{J^g}} }^*}(\mathcal B\,,\,\,0)} \right]$, 
which yield (\ref{eq_dupr}) and (\ref{eq_plowr}), respectively.


    \bibliographystyle{unsrt}    


\begin{thebibliography}{10}

\bibitem{meshkati_energy-efficient_2007}
F.~Meshkati, {H.V.} Poor, and {S.C.} Schwartz.
\newblock {Energy-Efficient} resource allocation in wireless networks.
\newblock {\em {IEEE} Signal Processing Magazine}, 24(3):58--68, 2007.

\bibitem{mingbo_xiao_utility-based_2003}
Mingbo Xiao, {N.B.} Shroff, and {E.K.P.} Chong.
\newblock A utility-based power-control scheme in wireless cellular systems.
\newblock {\em {IEEE/ACM} Transactions on Networking}, 11(2):210--221, 2003.

\bibitem{rey_robust_2005}
F.~Rey, M.~Lamarca, and G.~Vazquez.
\newblock Robust power allocation algorithms for {MIMO} {OFDM} systems with
  imperfect {CSI}.
\newblock {\em {IEEE} Transactions on Signal Processing}, 53(3):1070--1085,
  2005.

\bibitem{yingwei_yao_rate-maximizing_2005}
Yingwei Yao and {G.B.} Giannakis.
\newblock Rate-maximizing power allocation in {OFDM} based on partial channel
  knowledge.
\newblock {\em {IEEE} Transactions on Wireless Communications},
  4(3):1073--1083, 2005.

\bibitem{yoo_capacity_2006}
T.~Yoo and A.~Goldsmith.
\newblock Capacity and power allocation for fading {MIMO} channels with channel
  estimation error.
\newblock {\em {IEEE} Transactions on Information Theory}, 52(5):2203--2214,
  2006.

\bibitem{love_overview_2008}
{D.J.} Love, {R.W.} Heath, {V.K.N.} Lau, D.~Gesbert, {B.D.} Rao, and
  M.~Andrews.
\newblock An overview of limited feedback in wireless communication systems.
\newblock {\em {IEEE} Journal on Selected Areas in Communications},
  26(8):1341--1365, 2008.

\bibitem{bertsekas_dynamic_1987}
D.~P Bertsekas.
\newblock {\em Dynamic programming: deterministic and stochastic models}.
\newblock {Prentice-Hall,} Inc. Upper Saddle River, {NJ,} {USA}, 1987.

\bibitem{tassiulas_stability_1992}
L.~Tassiulas and A.~Ephremides.
\newblock Stability properties of constrained queueing systems and scheduling
  policies for maximum throughput in multihop radio networks.
\newblock {\em {IEEE} Transactions on Automatic Control}, 37(12):1936--1948,
  1992.

\bibitem{georgiadis_resource_2006}
L.~Georgiadis, M.~J Neely, and L.~Tassiulas.
\newblock {\em Resource allocation and cross layer control in wireless
  networks}.
\newblock Now Pub, 2006.

\bibitem{neely_energy_2006}
M.~J. Neely.
\newblock Energy optimal control for time-varying wireless networks.
\newblock {\em {IEEE} Transactions on Information Theory}, 52(7):2915–2934,
  2006.

\bibitem{berry_optimal_2006}
Randall Berry.
\newblock Optimal power-delay trade-offs in fading channels: Small delay
  asymptotics.
\newblock {\em Information Theory and {Applications—Inaugural} Workshop},
  2006.

\bibitem{qiao_miser:_2003}
Daji Qiao, Sunghyun Choi, Amit Jain, and Kang~G Shin.
\newblock {MiSer:} an optimal low-energy transmission strategy for {IEEE}
  802.11a/h.
\newblock In {\em Proceedings of the 9th annual international conference on
  Mobile computing and networking}, {MobiCom} '03, page 161–175, New York,
  {NY,} {USA}, 2003. {ACM}.
\newblock {ACM} {ID:} 939003.

\bibitem{zafer_minimum_2009}
M.~Zafer and E.~Modiano.
\newblock Minimum energy transmission over a wireless channel with deadline and
  power constraints.
\newblock {\em {IEEE} Transactions on Automatic Control}, 54(12):2841--2852,
  2009.

\bibitem{kushner_numerical_2001}
Harold~Joseph Kushner and Paul Dupuis.
\newblock {\em Numerical methods for stochastic control problems in continuous
  time}.
\newblock Springer, 2001.

\bibitem{bertsekas_dynamic_2007}
{D.P.} Bertsekas.
\newblock Dynamic programming and optimal control.
\newblock {\em Athena Scientific}, 2007.

\bibitem{zafer_calculus_2005}
M.~A Zafer and E.~Modiano.
\newblock A calculus approach to minimum energy transmission policies with
  quality of service guarantees.
\newblock In {\em {IEEE} {INFOCOM}}, volume~1, page 548, 2005.

\bibitem{marzetta_fast_2006}
{T.L.} Marzetta and {B.M.} Hochwald.
\newblock Fast transfer of channel state information in wireless systems.
\newblock {\em {IEEE} Transactions on Signal Processing}, 54(4):1268--1278,
  2006.

\bibitem{csi_error_it}
T.~Yoo and A.~Goldsmith.
\newblock {Capacity and power allocation for fading MIMO channels with channel
  estimation error}.
\newblock {\em {IEEE} Transactions on Information Theory}, 52(5):2203--2214,
  2006.

\bibitem{Ramyadelayfb2009}
T.~{R. Ramya} and S.~Bhashyam.
\newblock {Eigen-Beamforming with Delayed Feedback and Channel Prediction}.
\newblock In {\em Proc. ISIT}, June-July 2009.

\bibitem{BertsekasNeuro:1996}
D.~{P. Bertsekas} and J.~{N. Tsitsiklis}.
\newblock {\em Neuro-Dynamic Programming}.
\newblock Athena Scientifics, 1st edition, 1996.

\bibitem{Borkarbook:2008}
V.~{S. Borkar}.
\newblock {\em Stochastic Approximation: A Dynamical Systems Viewpoints}.
\newblock Cambridge University Press, United Kingdom, 1st edition, 2008.

\bibitem{lau_asymptotic_2008}
V.~Lau, Wing~Kwan Ng, and {D.S.W.} Hui.
\newblock Asymptotic tradeoff between cross-layer goodput gain and outage
  diversity in {OFDMA} systems with slow fading and delayed {CSIT}.
\newblock {\em {IEEE} Transactions on Wireless Communications},
  7(7):2732--2739, 2008.

\bibitem{wimax:2008}
IEEE 802.16m evaluation~methodology document.
\newblock IEEE 802.16m-08/004r4.

\bibitem{Proakis:2001}
J.~G. Proakis.
\newblock {\em {Digital Communications}}.
\newblock New York: McGraw-Hill, 4th ed.,, 2001.

\bibitem{berry_communication_2002}
R.~A. Berry and R.~G. Gallager.
\newblock Communication over fading channels with delay constraints.
\newblock {\em {IEEE} Transactions on Information Theory}, 48(5):1135–1149,
  2002.

\bibitem{shuguang_cui_energy-constrained_2005}
Shuguang Cui, {A.J.} Goldsmith, and A.~Bahai.
\newblock Energy-constrained modulation optimization.
\newblock {\em {IEEE} Transactions on Wireless Communications},
  4(5):2349--2360, 2005.

\bibitem{andrews_scheduling_2004}
M.~Andrews, K.~Kumaran, K.~Ramanan, A.~Stolyar, R.~Vijayakumar, and P.~Whiting.
\newblock Scheduling in a queuing system with asynchronously varying service
  rates.
\newblock {\em Probability in the Engineering and Informational Sciences},
  18(02):191–217, 2004.

\bibitem{eunkyung_kim_efficient_2007}
Eunkyung Kim, Juhee Kim, and Kyung~Soo Kim.
\newblock An efficient resource allocation for {TCP} services in {IEEE} 802.16
  wireless {MANs}.
\newblock In {\em Vehicular Technology Conference, 2007. {VTC-2007} Fall. 2007
  {IEEE} 66th}, pages 1513--1517, 2007.

\bibitem{narbutt_gauging_2006}
Miroslaw Narbutt and Mark Davis.
\newblock Gauging {VoIP} call quality from 802.11 {WLAN} resource usage.
\newblock In {\em A World of Wireless, Mobile and Multimedia Networks,
  International Symposium on}, volume~0, pages 315--324, Los Alamitos, {CA,}
  {USA}, 2006. {IEEE} Computer Society.

\bibitem{Ross:2003}
S.~M. Ross.
\newblock {\em {Introduction to probability models}}.
\newblock 8th edition, Amsterdam : Academic Press, 2003.

\end{thebibliography}

\end{document}